\documentclass[runningheads]{llncs}
\usepackage[T1]{fontenc}
\usepackage{lineno}

\usepackage{graphicx}
\usepackage{todonotes}
\usepackage{amsfonts}  
\usepackage{amssymb}
\usepackage{tikz}
\usepackage{circuitikz}
\usetikzlibrary{automata, positioning}
\usepackage{amsmath} 
\usepackage{hyperref}
\newtheorem{observation}{Observation}
\newtheorem{remark*}{Remark}
\usetikzlibrary{calc} 
\usetikzlibrary{shapes}
\usetikzlibrary{shapes.geometric, arrows.meta, positioning}
\usetikzlibrary{arrows.meta}
\usepackage{comment}
\usepackage{algorithm}
\usepackage{algorithmicx}
\usepackage[noend]{algpseudocode}
\usepackage{color}
\usepackage{todonotes}
\usepackage{bm}
\newcounter{rqrmnt}
 
\definecolor{darkgreen}{RGB}{0,140,0}
\newtheorem{hypothesis*}{Hypothesis}

\newcommand{\both}[1]{{(resp. #1)}}
\newcommand{\mvc}{\mathrm{mvc}}
\newcommand{\evc}{\mathrm{evc}}

\begin{document}

\title{Eternal Vertex Cover Problem on Halin Graphs}

\author{Jasine Babu
\and
Pratik Ghosal
\and
Cipriyano Simoes
}

\authorrunning{J. Babu, P. Ghosal and C. Simoes}

\institute{Indian Institute of Technology Palakkad, India
\email{ \{jasine@,pratik@,112204005@smail.\}iitpkd.ac.in}}
\maketitle

Eternal vertex cover problem is a graph protection problem which is a dynamic two player game variant of the classical vertex cover problem.
In this game, the minimum number of guards required to protect a graph $G$ is called the eternal vertex cover number of $G$, denoted by $\evc(G)$. Deciding whether $\evc(G)$ is at most $k$ or not is NP-hard even for bipartite graphs~\cite{misra2022eternal} and the complexity status remains open even for treewidth two graphs. For any graph $G$, $\mvc(G) \le \evc(G) \le 2\mvc(G)$, where $\mvc(G)$ is the vertex cover number of $G$. Tight examples are known for both bounds, though no biconnected graph $G$ can achieve $\evc(G)=2\mvc(G)$~\cite{Klostermeyer2009}. For infinite graph families $\mathcal{F}$, consider the parameter $\rho(\mathcal{F}) =   \sup \{r \in \mathbb{R} :\text{ for infinitely many graphs }G \in \mathcal{F}, \frac{\evc(G)}{\mvc(G)} \ge r \}$. For many classes of graphs such as series parallel graphs, bipartite graphs, and planar graphs, $\rho(\mathcal{F})=2$, even when restricted to their biconnected graphs~\cite{BABU22newlb}. On the other hand, for classes such as biconnected chordal graphs and locally connected graphs, $\rho$ is $1$.
No well-studied class of biconnected graphs $\mathcal{F}$ with unbounded vertex cover number is known yet, for which $1 < \rho(\mathcal{F}) < 2$.
Halin graphs are $3$-connected and they have treewidth three.  
In this paper, we show that when $\mathcal{F}$ is the family of Halin graphs, $\frac{7}{6} \le \rho(\mathcal{F}) \le \frac{3}{2}$.
For the upper bound, we give a defense algorithm that works for any Halin graph $G$, using at most $\frac{3}{2}\mvc(G)$ guards and just two configurations. This strategy is shown to require at most $\frac{4}{3}\mvc(G)$ guards for several subfamilies of Halin graphs.  
We give another improved algorithm that works for any caterpillar Halin graph $G$ using at most $\frac{4}{3}\mvc(G)$ guards. 
Our algorithms are linear time and they provide factor $\rho$ approximations for computing eternal vertex cover number for graphs in the respective classes.  
The lower bound of $\frac{7}{6}$ is also constructive and is achieved for a subfamily of caterpillar Halin graphs. 
The contrast in the value of $\rho$ for Halin graphs with that of series parallel graphs, which are precisely the class of graphs of treewidth at most two, is noteworthy.


\section{Introduction}\label{sec:intro}
Graph protection problems are studied extensively~\cite{goddard2005,cockayne2004roman,klostermeyer2017eternal,Hartnell2014} . 
One such problem is the eternal vertex cover problem, introduced by Klostermeyer and Mynhardt in 2009~\cite{Klostermeyer2009}. 
The problem is posed as a  multi-round attacker-defender game, which is a dynamic extension of the classical vertex cover problem~\cite{Karp72}. 

Let $G=(V, E)$ be a finite, simple, undirected graph. A \emph{vertex cover} of $G$ is a subset $S \subseteq V$ such that every edge $e\in E$ is incident with some vertex $v \in S$. A \emph{minimum vertex cover} of $G$ is a vertex cover of $G$ of minimum cardinality, and its size is the \emph{vertex cover number} of $G$, which is denoted by $\mvc(G)$. Computing vertex cover number is a well-known NP-hard problem~\cite{Karp72}, but it admits polynomial time algorithms for certain graph classes such as bipartite graphs~\cite[Chapter 1]{Vazirani} and constant treewidth graphs~\cite{Bodlaender1996tw}. 

The \emph{eternal vertex cover problem} considers a dynamic setting in which guards are placed on vertices and must defend 
against an infinite sequence of adversarial edge attacks. We consider the standard model in which at most one guard is allowed on a vertex at any instant. A \emph{configuration} is a vertex cover $C\subseteq V$ representing the current positions of the guards. The game begins with an initial configuration $C_0$, chosen by the defender. At each time step $i\geq 0$, given a configuration $C_i$, the adversary selects an edge $uv\in E$ to attack. The defender must move a guard located at $u$ or $v$ along the edge $uv$ to defend the attack. Simultaneously, every other guard can either be moved to an adjacent vertex or be retained in its position.
If both endpoints of an edge have guards, they both are allowed to move across the edge simultaneously.
The resulting set of occupied vertices forms the next configuration $C_{i+1}$, which must again be a vertex cover of $G$. If the defender is not able to execute such a reconfiguration, the adversary wins. Otherwise, the game proceeds to the next round with the adversary choosing an edge to attack. The number of guards remains the same throughout the game. If the defender can successfully defend any infinite sequence of attacks on $G$ using $k$ guards, we say that $G$ is $k$-defendable, and the configurations used in the game are said to form an eternal vertex cover class (\texttt{EVCC}) of $G$.  The \emph{eternal vertex cover number} of $G$, denoted by $\evc(G)$, is the minimum $k$ such that $G$ is $k$-defendable. 
For a cycle $C_n$ on $n$ vertices, $\mvc(C_n)=\evc(C_n)=\lceil\frac{n}{2}\rceil$.
For a path $P_n$ on $n$ vertices, $\mvc(P_n)=\lfloor\frac{n}{2}\rfloor$, whereas $\evc(P_n)=n-1$.  
For a tree $T$, $\evc(T)$ is one more than the number of internal vertices of $T$~\cite{Klostermeyer2009}.

The \emph{eternal vertex cover problem} takes any graph $G$ as input and has to compute $\evc(G)$.  The problem is NP-hard and $2$-approximable in polynomial time~\cite{Fomin2010}. The hardness holds even for bipartite graphs~\cite{misra2022eternal} and restricted classes of planar graphs~\cite{BABU2021}. For  chordal graphs~\cite{BPS2021}, cactus graphs~\cite{BPS2021} and generalized trees\cite{Hisashi_Gen_Trees}, the problem admits polynomial time algorithms. The status of the problem remains open for constant treewidth graphs~\cite{Fomin2010} and some of its subclasses like series parallel graphs~\cite{TizianaseriesParallel}, outerplanar graphs~\cite{BabuKPW25} and Halin graphs. 

It is known that for any graph $G$, $\mvc(G)\leq \evc(G)\leq 2\mvc(G)$~\cite{Fomin2010,Klostermeyer2009}. 
The class of graphs which achieve  $\evc(G)=2\mvc(G)$ was characterized by Klostermeyer et al.~\cite{Klostermeyer2009}. Interestingly, all such graphs are series parallel graphs and none of them is biconnected~\cite{Klostermeyer2009,BABU22newlb}. Further, no better lower bounds are known for biconnected graphs, even though $\evc(G)$ is lower bounded by the minimum size of a vertex cover that contains all cut vertices of $G$~\cite{BABU22newlb}. 
This makes the pursuit for better lower and upper bounds for $\evc(G)$ of biconnected graphs interesting. 

For all graph classes studied so far, either every biconnected graph in the family satisfies $\evc(G)\le \mvc(G)+1$ or there exist infinitely many biconnected graphs $G$ in the family for which $\frac{evc(G)}{mvc(G)}$ is greater than $2 - \epsilon$, for all $\epsilon >0$. Classes such as chordal graphs and locally connected graphs fall in the former category~\cite{BABU2021}; while bipartite, planar and series parallel graphs fall in the latter. In this work, for infinite families of biconnected graphs $\mathcal{F}$, we focus on the parameter $\rho(\mathcal{F}) =   \sup \{r \in \mathbb{R} :\text{ for infinitely many graphs }G \in \mathcal{F}, \frac{\evc(G)}{\mvc(G)} \ge r \}$.

An interesting question is whether there are classes of biconnected graphs $\mathcal{F}$ with unbounded vertex cover number, for which $1 < \rho(\mathcal{F}) < 2$. We answer this question in the affirmative, by showing that for the family of Halin graphs $\mathcal{H}$, which is a well studied~\cite{Halin1971orig,Syslo1979halin,Cornuejols1983,Eppstein1988} graph class, $\frac{7}{6} \le \rho(\mathcal{H}) \le \frac{3}{2}$.  A Halin graph is a graph constructible from a tree $T$ with no degree two vertex, by connecting all  leaves of $T$ to form a cycle $C$ that respects the order induced by a planar embedding of $T$~\cite{Halin1971orig}.
Linear time algorithms are known for the recognition of Halin graphs and decomposition of a Halin graph into the underlying tree and cycle ~\cite{Eppstein2016,FominT06}.
Halin graphs are $3$-connected planar graphs~\cite[Chapter H]{Hazewinkel1997}, Hamiltonian~\cite{Syslo1979halin} and have treewidth exactly three~\cite{Bodlaender98}. 

 To establish the upper bound $\rho(\mathcal{H})\le \frac{3}{2}$, we give an algorithm that, given any $n$ vertex Halin graph $G$ as input,  constructs an \texttt{EVCC} of $G$  with only two configurations, each of size at most $\frac{3n}{4} \le \frac{3}{2}\mvc(G)$. 
A tighter analysis of this algorithm shows that for all Halin graphs except those having $\lceil \frac{n}{2} \rceil \le \mvc(G) < \frac{9n}{16}$, the factor $\frac{3}{2}$ can be improved to $\frac{4}{3}$. The ratio $\frac{4}{3}$ also works if the number of leaves in the underlying tree of $G$ is at least $\frac{2n}{3}$. 
A caterpillar is a tree with a dominating path. 
A Halin graph whose underlying tree is a caterpillar is called a caterpillar Halin graph. For some families of caterpillar Halin graphs, our first algorithm can only achieve the bound $\frac{3}{2}\mvc(G)$. 
We give another algorithm that achieves the ratio $\frac{4}{3}$ for caterpillar Halin graphs. 
For the lower bound on $ \rho(\mathcal{H})$, we give an infinite subfamily of caterpillar Halin graphs $\mathcal{H'}$ for which $\rho(\mathcal{H}') \ge \frac{7}{6}$.

Our algorithms are linear time and they serve as $\rho$ factor approximation algorithms for computing $\evc(G)$ for graphs $G$ in the respective classes. 
It is interesting to note that Halin graphs have treewidth $3$ and yet $\rho(\mathcal{H})\le \frac{3}{2}$, whereas for biconnected series parallel graphs, which are precisely the class of graphs with treewidth at most two, $\rho(\mathcal{F})=2$. 
The complexity status of the eternal vertex cover problem for Halin graphs as well as series parallel graphs remains open.

\section{Notation}\label{sec:notation}
We consider only simple, finite and undirected graphs. The family of Halin graphs will be denoted as $\mathcal{H}$. 
Throughout the paper, a Halin graph $G$ will be viewed as the union of its underlying tree $T$ and its outer cycle $C$. We use $n$ to denote $|V(G)|$, $l$ to denote $|V(C)|$, $L(T)$ to denote the set of leaf vertices of $T$ and $int(T)$ to denote the set of internal vertices of $T$.
Unless stated otherwise, the vertices of the cycle are named in
anticlockwise order as
$c_0,c_1,\ldots,c_{l-1}$.
Whenever an index is outside this range, it is interpreted modulo $l$. For any $v \in int(T)$ which is adjacent to at most one other internal vertex, the induced subgraph of $G$ on the vertex set consisting of $v$ and its neighbours in $L(T)$ is called a \emph{fan}, and the vertex $v$ is called the
\emph{fan center}. We repeatedly use the fact that if $|int(T)|>1$ then $G$ has at least two
fans. 
Proofs omitted from the main text are available in the Appendix.

\section{Lower Bound for $\rho(\mathcal{H})$} \label{sec:lowerbound}

In this section, we describe a family of Halin graphs for which the eternal vertex cover number is larger than the vertex cover number by more than a constant additive gap. The construction is based on repeating a fixed local structure, which we call a unit.

A unit graph $U_i$ on $6$ vertices is as shown in Figure~\ref{fig:G1_G2}. The Halin graph $G_r$ is constructed with disjoint copies of $U_i$, $1 \leq i \leq r$ and three additional vertices $x,y$ and $z$. Other than the edges inside each unit, the additional edges in $G_r$ are $xy,yz,xz,za_{1,2},ya_{1,1},za_{r,6}$ and for each $1 \leq i \leq r-1$, the edges  $a_{i,3}a_{i+1,2}$ and $a_{i,6}a_{i+1,1}$. Figure~\ref{fig:G1_G2} shows the graph $G_r$. Observe that
$\mvc(G_1)=5$ and $\evc(G_1)=6$. An \texttt{EVCC} of $G_1$ with six guards is formed by
the configurations
$C_1=\{x,z,a_{1,1},a_{1,2},a_{1,3},a_{1,5}\},$ and $C_2=\{y,z,a_{1,2},a_{1,3},a_{1,4},a_{1,6}\}$.

\begin{figure}[ht]
    
    \resizebox{0.9\linewidth}{!}{%
    \tiny
   \begin{minipage}{0.30\linewidth}
\scriptsize
\begin{tikzpicture}[scale=0.8,
    label distance= -0.11cm,
    every node/.style={circle,fill=black,inner sep=2pt}
    ]
\node[label={[yshift=-3pt]above:$a_{i,2}$}] (a12) at (2.5,1.5) {};
\node[label=below:$a_{i,1}$] (a11) at (2.5,0) {};
\node[label=below:$a_{i,6}$] (a16) at (6.5,0) {};
\node[label={[yshift=-3pt]above:$a_{i,3}$}] (a13) at (5,1.5) {};
\node[label=below:$a_{i,5}$] (a15) at (5,0) {};
\node[label=below:$a_{i,4}$] (a14) at (3.5,0) {};

\draw (a12)--(a13);
\draw (a12)--(a11);
\draw (a11)--(a14);
\draw (a13)--(a16);
\draw (a13)--(a15);
\draw (a13)--(a14);

\draw (a16)--(a15);
\draw (a15)--(a14);
\end{tikzpicture} 
\vspace{-0.2cm}
\end{minipage}%
\hspace{0.1\linewidth}%
\begin{minipage}{0.58\linewidth}
\scriptsize
\begin{tikzpicture}[scale=0.8,
    label distance= -0.1cm,
    every node/.style={circle,fill=black,inner sep=2pt}]

\node[label=above:$x$] (x) at (0,2) {};
\node[label=below:$y$] (y) at (0,0) {};

\node[label=above:$z$] (z) at (1,1) {};

\node[label={[yshift=-3pt]above:$a_{1,2}$}] (a12) at (2,1) {};
\node[label=below:$a_{1,1}$] (a11) at (2,0) {};

\node[label={[yshift=-3pt]above:$a_{1,3}$}] (a13) at (4,1) {};
\node[label=below:$a_{1,4}$] (a14) at (3,0) {};
\node[label=below:$a_{1,5}$] (a15) at (4,0) {};
\node[label=below:$a_{1,6}$] (a16) at (5,0) {};

\node[label={[yshift=-3pt]above:$a_{r,2}$}] (a22) at (6.5,1) {};
\node[label=below:$a_{r,1}$] (a21) at (6.5,0) {};

\node[label={[yshift=-3pt]above:$a_{r,3}$}] (a23) at (7.7,1) {};
\node[label=right:$a_{r,6}$] (a26) at (9,2) {};
\node[label=right:$a_{r,5}$] (a25) at (9,1) {};
\node[label=right:$a_{r,4}$] (a24) at (9,0) {};

\draw (x)--(y);
\draw (x)--(a26);
\draw (y)--(a11);

\draw (x)--(z);
\draw (y)--(z);

\draw (z)--(a12);
\draw (a12)--(a11);

\draw (a12)--(a13);
\draw (a11)--(a14);
\draw (a14)--(a15);
\draw (a15)--(a16);

\draw (a13)--(a14);
\draw (a13)--(a15);
\draw (a13)--(a16);

\draw[dashed] (a13)--(a22);
\draw[dashed] (a16)--(a21);
\draw (a22)--(a21);
\draw (a21)--(a24);

\draw (a22)--(a23);

\draw (a23)--(a26);
\draw (a23)--(a25);
\draw (a23)--(a24);

\draw (a26)--(a25);
\draw (a25)--(a24);
\end{tikzpicture}
\vspace{-0.5cm}
\end{minipage}%
}
    \caption{The unit graph $U_i$ (left) and the Halin graph $G_r$ (right).}
    \label{fig:G1_G2}
\end{figure}
For a unit $U_i$, let $\eta_X(U_i)$ denote the number of guards placed on the vertices of $U_i$ in a configuration $X$. We omit the subscript when it is clear from the context. It is easy to see that in any vertex cover of $G_r$, $\eta(U_i) \ge 3$.

\begin{proposition}\label{prop:unit-mvc}
$\mvc(G_r)= 3r+2$. Moreover, the unique minimum vertex cover of $G_r$ has exactly three vertices from each unit.
\end{proposition}
  The family $\{ G_r \mid r \ge 1 \}$ is used to lower bound $\rho(\mathcal{H})$.
 The following lemma is crucial to prove a lower bound for $\evc(G_r)$. 
\begin{lemma}\label{lem:four_units}
Let $r \ge 4$ and $X$ be an eternal vertex cover configuration of $G_r$. Let $U_a,U_b,U_c$ and $U_d$ be four consecutive units of $G_r$. If $ \eta(U_a)=\eta(U_b) =3$ in $X$, then either
(i) $ \eta(U_{c}) \ge 5 $ or  (ii) 
$\eta(U_{c}) = 4 \text{ and } \eta(U_{d}) \ge 4$.
\end{lemma} 
\begin{proof}
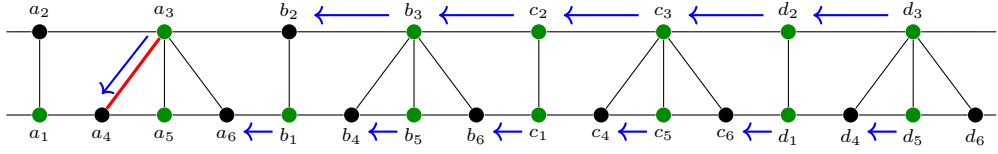
\begin{figure}[h]
\centering
\resizebox{\linewidth}{!}{%
\scriptsize
\begin{tikzpicture}[scale=1.1,
    label distance= -0.1cm,
    every node/.style={circle,fill=black,inner sep=2pt},
    attacked/.style={red, very thick},
    move/.style={blue, thick, ->}]

\node[fill=black,label=above:$a_2$] (a2) at (0,1) {};
\node[fill=darkgreen,label=above:$a_3$] (a3) at (1.5,1) {};
\node[fill=black,label=above:$b_2$] (b2) at (3,1) {};
\node[fill=darkgreen,label=above:$b_3$] (b3) at (4.5,1) {};
\node[fill=darkgreen,label=above:$c_2$] (c2) at (6,1) {};
\node[fill=darkgreen,label=above:$c_3$] (c3) at (7.5,1) {};
\node[fill=darkgreen,label=above:$d_2$] (d2) at (9,1) {};
\node[fill=darkgreen,label=above:$d_3$] (d3) at (10.5,1) {};

\node[fill=darkgreen,label=below:$a_1$] (a1) at (0,0) {};
\node[fill=black,label=below:$a_4$] (a4) at (0.75,0) {};
\node[fill=darkgreen,label=below:$a_5$] (a5) at (1.5,0) {};
\node[fill=black,label=below:$a_6$] (a6) at (2.25,0) {};

\node[fill=darkgreen,label=below:$b_1$] (b1) at (3,0) {};
\node[fill=black,label=below:$b_4$] (b4) at (3.75,0) {};
\node[fill=darkgreen,label=below:$b_5$] (b5) at (4.5,0) {};
\node[fill=black,label=below:$b_6$] (b6) at (5.25,0) {};

\node[fill=darkgreen,label=below:$c_1$] (c1) at (6,0) {};
\node[fill=black,label=below:$c_4$] (c4) at (6.75,0) {};
\node[fill=darkgreen,label=below:$c_5$] (c5) at (7.5,0) {};
\node[fill=black,label=below:$c_6$] (c6) at (8.25,0) {};

\node[fill=darkgreen,label=below:$d_1$] (d1) at (9,0) {};
\node[fill=black,label=below:$d_4$] (d4) at (9.75,0) {};
\node[fill=darkgreen,label=below:$d_5$] (d5) at (10.5,0) {};
\node[fill=black,label=below:$d_6$] (d6) at (11.25,0) {};


\draw (-0.4,1)--(a2);
\draw (a2)--(a3);
\draw (a3)--(b2);
\draw (b2)--(b3);
\draw (b3)--(c2);
\draw (c2)--(c3);
\draw (c3)--(d2);
\draw (d2)--(d3);
\draw (d3)--(11.5,1);
\draw (-0.4,0)--(a1);
\draw (a1)--(a4);
\draw (a4)--(a5);
\draw (a5)--(a6);
\draw (a6)--(b1);
\draw (b1)--(b4);
\draw (b4)--(b5);
\draw (b5)--(b6);
\draw (b6)--(c1);
\draw (c1)--(c4);
\draw (c4)--(c5);
\draw (c5)--(c6);
\draw (c6)--(d1);
\draw (d1)--(d4);
\draw (d4)--(d5);
\draw (d5)--(d6);
\draw (d6)--(11.5,0);

\draw (a1)--(a2);
\draw (a3)--(a4);
\draw (a3)--(a5);
\draw (a3)--(a6);

\draw (b1)--(b2);
\draw (b3)--(b4);
\draw (b3)--(b5);
\draw (b3)--(b6);

\draw (c1)--(c2);
\draw (c3)--(c4);
\draw (c3)--(c5);
\draw (c3)--(c6);
\draw (d1)--(d2);
\draw (d3)--(d4);
\draw (d3)--(d5);
\draw (d3)--(d6);

\draw[attacked] (a3)--(a4);


\draw[move] ($(a3)+(-0.2,-0.05)$) -- ($(a4)+(0,0.25)$);
\draw[move] ($(b1)+(-0.2,-0.2)$) -- ($(a6)+(0.2,-0.2)$);
\draw[move] ($(b5)+(-0.2,-0.2)$) -- ($(b4)+(0.2,-0.2)$);
\draw[move] ($(c1)+(-0.2,-0.2)$) -- ($(b6)+(0.2,-0.2)$);
\draw[move] ($(c5)+(-0.2,-0.2)$) -- ($(c4)+(0.2,-0.2)$);
\draw[move] ($(d1)+(-0.2,-0.2)$) -- ($(c6)+(0.2,-0.2)$);
\draw[move] ($(d5)+(-0.2,-0.2)$) -- ($(d4)+(0.2,-0.2)$);

\draw[move] ($(b3)+(-0.3,0.2)$) -- ($(b2)+(0.3,0.2)$);
\draw[move] ($(c2)+(-0.3,0.2)$) -- ($(b3)+(0.3,0.2)$);
\draw[move] ($(c3)+(-0.3,0.2)$) -- ($(c2)+(0.3,0.2)$);
\draw[move] ($(d2)+(-0.3,0.2)$) -- ($(c3)+(0.3,0.2)$);
\draw[move] ($(d3)+(-0.3,0.2)$) -- ($(d2)+(0.3,0.2)$);

\end{tikzpicture}
}
\vspace{-0.5cm}
\caption{The guard movements forced in response to an attack on the edge $a_3a_4$. }
\label{fig:four-units}

\end{figure}
Suppose $\eta(U_a)=\eta(U_{b})=3$.
Then the set of occupied vertices in these two units can only be 
$\{a_1,a_3,a_5,b_1,b_3,b_5\}$, as shown in Figure~\ref{fig:four-units}.
 Now, consider an attack on the edge $a_3a_4$. To defend the attack, the guard from $a_3$ must move to $a_4$. Further, the resultant configuration should be a vertex cover. Since the triangle  $a_3a_5a_6$ requires at least two of its vertices to have guards, the guard on $b_1$ is forced to move to $a_6$.
 Now to defend the edge $b_1b_2$ in the resultant configuration, the guard on $b_3$  must move to $b_2$.
 To defend edge $b_1b_4$, the guard from $b_5$ has to move to $b_4$.
 Since the triangle $b_3b_5b_6$ requires at least two of its vertices to have guards, and both the guards which were previously on it were forced to move out, the two guards moving in must be from $c_1$ and $c_2$. This implies that $c_1,c_2 \in X$. Hence, $\eta(U_c) \ge 4$ in  configuration $X$.
 
 Now if $\eta(U_c) >4$, then we are done. So assume that $\eta(U_c) \le 4$. Then $X$ is a safe configuration only when the occupied vertices are $c_1,c_2,c_3$ and $c_5$.
 Since, the guards on $c_1$ and $c_2$ are forced to move to $b_6$ and $b_3$, respectively, to protect the edge $c_1c_2$, the guard on $c_3$ must move to $c_2$. Further, to protect the edge $c_1c_4$, the guard on $c_5$ must move to $c_4$.
 Now, to protect triangle $c_3c_5c_6$, two guards must move in from $U_d$.
  Hence, the configuration $X$ should be such that $d_1$ and $d_2$ both have guards. This implies $\eta(U_d) \ge 4$.  Hence, it is the case that $\eta(U_c) =4$ and $\eta(U_d) \ge 4$. This concludes the proof.
\end{proof}
\begin{definition}[Block]
    Let $X$ be an eternal vertex cover configuration of $G_r$. A block of $G_r$ is any maximal sequence of consecutive units $U_i,U_{i+1}, \ldots, U_j$ of $G_r$, in which there is an index $k \in [i,j]$ such that
        (i) $\forall t \in [i,k]$, $\eta(U_t) =3$    and
        (ii) $\forall  t \in [k+1,j]$, $\eta(U_t) \ge 4$. 
        If all units $U_i$ of $G_r$ have $\eta(U_i) >3$ then $U_1, U_2, \ldots, U_r$ is the unique block of $G_r$. 
\end{definition}
For the rest of this section, we assume an arbitrary eternal vertex cover configuration of $G_r$.
Let $U_i,U_{i+1}, \ldots, U_j$ be a block  of $G_r$. From the maximality of a block, it follows that if $i>1$, then $\eta(U_{i-1}) \ge 4$ and if $j<r$, then $\eta(U_{j+1}) =3$. 
By the definition of a block and Lemma~\ref{lem:four_units}, the following observation is immediate. 
\begin{observation}\label{obs:block_prop}
 Every unit belongs to some block and a block has at most two units with exactly three guards.   
\end{observation}
From Observation~\ref{obs:block_prop}, we can show the following.
\begin{lemma}\label{lem:block_notlast}
    In each block $B = U_i,U_{i+1} \ldots, U_{j}$, except the last block, the average number of guards per unit is at least $3.5$. 
\end{lemma}
Observation~\ref{obs:block_prop} also gives the following bound for the last block.
\begin{lemma} \label{lem:block_last}
    If $B = U_i,U_{i+1} \ldots, U_{r}$ is the last block and it has $m$ units, then the total number of guards in $B$ is at least $3.5m - 1$.
\end{lemma}

Now, from Lemma~\ref{lem:block_notlast} and Lemma~\ref{lem:block_last}, we derive the following main theorem of this section.
\begin{theorem}\label{thm:lower_bound_family}
$evc(G_r)\ge 3.5r+1$. Consequently, for all $\epsilon>0$, there exists $r_0 \in \mathbb{N}$ such that for every $r > r_0$, $\frac{\evc(G_r)}{\mvc(G_r)} \ge \frac{7}{6}- \epsilon$ and $\rho(\mathcal{H}) \ge \frac{7}{6}$.
\end{theorem}

\section{General Bounds for Eternal Vertex Covers in Halin Graphs}\label{sec:generalbounds}
In this section, we propose a linear time algorithm that computes an eternal vertex cover class of size at most $\frac{3n}{4}$, for any Halin graph $G$ on $n$ vertices. Using the known lower bound for the vertex cover number of Halin graphs, we show that the size of the \texttt{EVCC} computed is at most 1.5 times the vertex cover number of $G$.
The algorithm is stated below. 
\begin{algorithm*}
\caption{Computing an \texttt{EVCC} of size at most $\frac{3n}{4}$}
\label{algo:evchalin11} 
\textbf{Input:} A Halin graph $G=(V,E)$. \\
\textbf{Output:} An \texttt{EVCC}
$\mathcal{C}=\{C_1,C_2\}$ of $G$ and its size $k$, where
$k\le \frac{3n}{4}$. 
\begin{algorithmic}
\State Compute the decomposition of $G$ into the underlying tree $T$ and the cycle $C$.
\State Order the vertices of $C$ as
$v=c_0,c_1,\ldots,c_{l-1},c_0$ in anticlockwise order.
\State $C_1 \gets int(T)\cup\{c_i \mid 0\le i\le l-1,\ i\text{ is even}\}$.
\State $C_2 \gets int(T)\cup\{c_{i+1\bmod l}\mid 0\le i\le l-1,\ i\text{ is even}\}$.
\State $\mathcal{C}\gets\{C_1,C_2\}$ and $k\gets \left\lceil\frac{l}{2}\right\rceil+|int(T)|$.
\State \Return $\mathcal{C},k$.
\end{algorithmic}
\end{algorithm*}

Since every Halin graph is Hamiltonian~\cite{Syslo1979halin}, it contains a cycle passing through all $n$ vertices. 
Therefore, any vertex cover must cover the edges of this Hamiltonian cycle, and hence we have the following lower bound for the vertex cover number of Halin graphs.
\begin{observation}\label{obs:mvchalin}
For any Halin graph $G$, $\mvc(G)\geq \left\lceil \frac{n}{2}\right\rceil$ .  
\end{observation}

To explain the correctness of Algorithm~\ref{algo:evchalin11}, we define four basic guard movements on the vertices of the cycle. This definition will also be used extensively in the rest of this paper.

\begin{definition}\label{def:cyclemovements}
     Let $D$ be a configuration of $G$ and let the cycle vertices be numbered  as $c_0,c_1, \ldots , c_{l-1}$ in the anti-clockwise direction. 
     Let $\Psi$ be a subpath of $C$.
     We define four types of guard movement on $\Psi$ as follows:
     \begin{itemize}
  \item \textbf{Anticlockwise Shift:} For every $c_i$ in $\Psi \cap D$, the guard on $c_i$ moves to $c_{i+1}$ (indices mod $l$).
  \item \textbf{Restricted Anticlockwise Shift:} For every $c_i$ in $\Psi \cap D$, the guard on $c_i$ moves to $c_{i+1}$ if and only if $c_{i+1} \notin D$.
  \item \textbf{Clockwise Shift:}  For every $c_i$ in $\Psi \cap D$, the guard on $c_i$ moves to $c_{i-1}$.
  \item \textbf{Restricted Clockwise Shift:} For every $c_i$ in $\Psi \cap D$, the guard on $c_i$ moves to $c_{i-1}$ if and only if $c_{i-1} \notin D$.
\end{itemize}
     
\end{definition}
 The following observation is simple.
\begin{observation}\label{obs:fullcyclemove}
    Let $G$ be a Halin graph. Let $C_1,C_2$ be two  vertex covers of $G$ satisfying $C_2 \cap V(C) = \{ c_{i+1} \mid c_i \in C_1 \cap V(C) \}$.  Let $\Psi$ be the subpath $c_0,c_1, \ldots,c_{l-1}$ that contains all the vertices of $C$.
    \begin{itemize}
        \item     Starting from $C_1$, on applying either an anticlockwise shift on $\Psi$ or a restricted clockwise shift on $\Psi$, the resultant guard positions on $V(C)$ will be same as they are in the configuration $C_2$.
        \item     Starting from $C_2$, on applying either a clockwise shift on $\Psi$ or a restricted anticlockwise shift on $\Psi$, the resultant guard positions on $V(C)$ will be same as they are in the configuration $C_1$.
    \end{itemize}
 
\end{observation}

The following lemma shows the correctness of Algorithm~\ref{algo:evchalin11}.

\begin{lemma}\label{lem:nbound}
    For any given Halin graph $G$, the class $\mathcal{C}$ returned by Algorithm~\ref{algo:evchalin11} forms an eternal vertex cover class of $G$ and it is of size at most  $\frac{3n}{4}$.
\end{lemma}

Since a decomposition of a Halin graph $G$ into its underlying tree and the cycle can be computed in linear time~\cite{Eppstein2016}, the running time of Algorithm~\ref{algo:evchalin11} is linear. Now, combining Observation~\ref{obs:mvchalin} and Lemma~\ref{lem:nbound}, we have the following theorem.
\begin{theorem}\label{thm:bound1}
For every Halin graph $G$, 
$\left\lceil\frac{n}{2}\right\rceil  \le \mvc(G) \le \evc(G)\le \frac{3n}{4}$.
Further, an eternal vertex cover class of $G$ of size at most
$\frac{3}{2}\,\mvc(G)$ can be computed in linear time.
\end{theorem}

\begin{remark}
    When $G = K_4$, $\evc(G) = 3=\frac{3n}{4}$. Since $K_4$ is a Halin graph, the upper bound given in Theorem~\ref{thm:bound1} in terms of $n$ is tight. The upper bound in terms of $\mvc(G)$ for Halin graphs has scope to be improved. However, it may be noted that for general graphs $2\mvc(G)$ is the best upper bound possible for $\evc(G)$ in terms of $\mvc(G)$~\cite{Klostermeyer2009}.
\end{remark}

\subsection{Improved Bounds using Algorithm~\ref{algo:evchalin11}}

 We believe that the bounds given by Theorem~\ref{thm:bound1} for Halin graphs might not be tight. In the rest of this section,  we show that the upper bound obtained from Algorithm~\ref{algo:evchalin11}, given in Theorem~\ref{thm:bound1}, can be improved  to $\frac{4}{3}\mvc(G)$ for many subclasses of Halin graphs.
 \medskip \\
\emph{Halin Graphs with High Vertex Cover Number: }
Consider any Halin graph $G$ for which $\mvc(G) \geq \frac{9n}{16}$. By Theorem~\ref{thm:bound1},
$\evc(G) \leq \frac{3n}{4} \leq \frac{4}{3}\mvc(G) $. \medskip \\
\emph{Halin Graphs with Large Number of Leaves in its Tree: }
Let $G$ be any Halin graph with $T$ being its underlying tree. Let $l$ be the number of leaves in $T$.
Suppose $l \geq  2\left\lceil \frac{n}{3} \right\rceil$. By Lemma~\ref{lem:nbound}, we have $ \evc(G) \le \left\lceil \frac{l}{2} \right\rceil + |int(T)|$. 
Since, $|int(T)|= n-l$, this gives $\evc(G) \leq n - \lfloor \frac{l}{2} \rfloor $.
Substituting $l \geq 2\left\lceil \frac{n}{3} \right\rceil$,
we get $\evc(G) \leq \frac{2n}{3}$. 
By Theorem~\ref{thm:bound1}, we have $\mvc(G) \ge  \left\lceil\frac{n}{2} \right\rceil$, giving $\evc(G) \le \frac{4}{3} \mvc(G)$.

From these facts, we have the following theorem.
\begin{theorem}\label{thm:improvedboundalgo1}
    Let $G$ be an $n$ vertex Halin graph that satisfies any of the following conditions:
        (i) $\mvc(G) \ge \frac{9n}{16}$ or
        (ii) the number of leaves $l$ in the underlying tree of $G$ is at least  $2\left\lceil \frac{n}{3} \right\rceil$.
    Then $\evc(G) \leq \frac{4}{3}\mvc(G)$. Further, Algorithm~\ref{algo:evchalin11} outputs an \texttt{EVCC} of $G$ of size at most
$\frac{4}{3}\mvc(G)$ in linear time.
\end{theorem}

\begin{remark}\label{rmk:4_deg_halin_graph}
    Any Halin graph $G$ for which all vertices in $int(T)$ have degree more than $3$ satisfy condition (ii) of Theorem~\ref{thm:improvedboundalgo1} and hence have $\evc(G)\leq \frac{4}{3}\,\mvc(G)$.
\end{remark}

\section{Improved Bounds for Caterpillar Halin Graphs}
Recall that a caterpillar Halin graph is a Halin graph whose underlying tree is a caterpillar.
There exist many caterpillar Halin graphs for which Algorithm~\ref{algo:evchalin11} can only give an \texttt{EVCC} whose size is close to $\frac{3}{2}$ times its vertex cover number (See Figure~\ref{fig:bad_example_algo1} of Appendix). 
In this section, we propose another linear time algorithm that computes an \texttt{EVCC} of size at most $\frac{4}{3}\mvc(G)$, for any caterpillar Halin graph $G$.

Let $G$ be a caterpillar Halin graph of size $n$, and let $T$ its underlying caterpillar tree and $C$ be the outer cycle formed by the leaves of $T$. Let $|C|=l$. 
We will refer to the dominating path of the caterpillar as the spine of $G$. 
We use $P(T)$ to denote the vertex set of the spine.
Since we are only concerned with trees underlying Halin graphs, without loss of generality, assume that $P(T) \cap V(C) = \emptyset$.
With this assumption, the endpoints of the spine are centers of fans (fan-centers).  
If $|P(T)|=1$, then the graph is a wheel. 
Since wheel graphs are internally triangulated planar graphs, it follows from \cite{BABU2021} that their eternal vertex cover number coincides
with their vertex cover number. 
In the rest of this section, we assume that $G$ is not a wheel.

Our algorithm is based on the construction of an \texttt{EVCC} consisting of two carefully designed configurations. Starting from a suitable minimum vertex cover $S$ of $G$, we identify a subset of vertices in $P(T) \setminus S$ and place guards on them in addition to $S$, in the first configuration. This allows to choose a second configuration that ensures valid reconfigurations between the two configurations are possible while defending all possible attacks, satisfying the guard movement constraints. 

\begin{algorithm}[t]
    \caption{Computing an \texttt{EVCC} of a caterpillar Halin graph of size $\frac{4}{3}\mvc(G)$} \label{algo:main} 
\textbf{Input:} A Halin graph $G$ with underlying caterpillar tree $T$ and cycle $C$ \\ \Comment{Assuming $G$ is not a wheel.} \\
\textbf{Outputs:} An \texttt{EVCC} $\mathcal{C}= \{C_1,C_2\}$ of $G$ and its size $k$.
\begin{algorithmic}
\State Compute a minimum vertex cover $S$ of $G$ such that no three consecutive vertices of $C$ are in $S$. \Comment{ Can be done in linear time}
\State Identify the fan centers $r$ and $f$. Consider the spine to be rooted at $r$.
\State Order the vertices of $C$ as $c_0,c_1, \ldots, c_{l-1}$ in anticlockwise order ensuring that $c_0$ is a neighbour of $r$.
\State $H \gets P(T) \setminus S$ 
\State $F_1 \gets \{ v \mid v \in H$, $v \neq f,r$ and  $v$ has multiple neighbours on $C \}$
\State $F_2 \gets \{ v \mid v \in H$, $v \neq f$ and the child of $v$ in the spine has a neighbour in $S \cap V(C) \}$
\State  $F \gets F_1 \cup F_2$.
\State $P_1 \gets \left(P(T) \cap S\right) \cup F \cup \{f\}$, \hspace{0.4cm} $Q_1 \gets V(C) \cap S$.
 \State $C_1  \gets P_1 \cup Q_1$, \hspace{0.4cm} $P_2 \gets \{v \mid v \in P(T), ~\text{parent of }v~\text{is in }P_1\} \cup \{r\}$.
\State  $Q_2 \gets \{ c_{i+1 \mod{l}}\mid \, c_i \in C_1 \cap V(C) \}$.  \Comment{$c_{i+1}$ is anti-clockwise neighbour of $c_i$ in $C$}.
\State  $C_2 \gets P_2 \cup Q_2$. 
\State Let $\mathcal{C} \gets \{C_1,C_2\}$ and $k= |C_1|$.
\State Return $\mathcal{C}$ and $k$.
\end{algorithmic} 
\end{algorithm}
A minimum vertex cover of a
Halin graph can be computed in linear time~\cite{BorieParkerTovey2008}.
The following is an interesting observation. 
\begin{observation}
    Every Halin graph has minimum vertex cover $S$, 
such that no three consecutive vertices of its underlying cycle are in $S$.
\end{observation}
To see this, we may start any minimum vertex cover $S'$ of the graph and modify it.
If there are three such consecutive cycle vertices $x, y, z$ in  $S'$, it is sufficient to include the non-cycle neighbour of $y$ in $S$, instead of $y$.

For convenience of the proof of correctness of Algorithm~\ref{algo:main}, we consider a planar embedding of the graph in such a way that the spine of the caterpillar is horizontal. 
We treat the spine as a rooted path with the rightmost spine vertex treated as the root vertex $r$.  Note that $r$ is a fan center. Let $f$ be the other fan center in $G$.
We number the cycle vertices as $c_0,c_1,\dots,c_{l-1} $
starting from a cycle neighbour of $r$ and proceeding in the anti-clockwise direction along the cycle. Throughout, the subscripts on the cycle vertices are interpreted modulo $l$.
From the construction of $\mathcal{C}$ in Algorithm~\ref{algo:main}, the following observation is immediate.
\begin{observation}\label{obs:f_in_c1}
    For any given caterpillar Halin graph $G$, the class $\mathcal{C}= \{ C_1,C_2\}$ given by Algorithm~\ref{algo:main} is such that (i) $f \in C_1$ and (ii) $r \in C_2$.
\end{observation}
By the definitions of $P_1,Q_1,P_2$ and $Q_2$, we get the following lemma. 
\begin{lemma}\label{lem:neighbour_of_child_of_hole_in_C1}
    For any given caterpillar Halin graph $G$, the class $\mathcal{C}= \{ C_1,C_2\}$ given by Algorithm~\ref{algo:main} is such that 
    \begin{itemize}
        \item[i] For any $x \in P(T)\setminus C_1$, the child of $x$ in $P(T)$ has an unoccupied neighbour on $C$.
        \item[ii] For any $x \in P(T)\setminus C_2$, there exists an ancestor $x'$ of $x$ in $P(T)$ with an unoccupied neighbour on $C$.
    \end{itemize}
\end{lemma}

To show that class $\mathcal{C}$ is a valid \texttt{EVCC} of $G$, we prove the following lemma.

\begin{lemma}\label{lem:caterpillarproof}
    For any caterpillar Halin graph $G$, the class $\mathcal{C}=\{C_1,C_2\}$, returned by Algorithm~\ref{algo:main}, is an eternal vertex cover class of $G$.
\end{lemma}

The proof of Lemma~\ref{lem:caterpillarproof} uses some guard  movement rules on the spine vertices, in addition to those defined on cycle vertices in Definition~\ref{def:cyclemovements}. 

\begin{definition}\label{def:spinemovements}
    Let $G$ be a caterpillar Halin graph and let the spine of $G$ be rooted at $r$. Let $\Phi$ be a subpath of the spine.
    From a configuration $D$ of $G$, we define four types of guard movements on $\Phi$ as follows:
    \begin{itemize}
        \item \textbf{Hole-Child exchange:} For every $v \in V(\Phi) \setminus D $, if $v$ has a child $u$ in $P(T)$, the guard on $u$ moves to $v$ and  no guards move to $u$.
        \item \textbf{Shift to Child:} For every $v \in V(\Phi) \cap D$, if $v$ has a child $u \in P(T)$, then the guard on $v$ moves to $u$.
        \item \textbf{Hole-Parent exchange:} For every $v \in V(\Phi) \setminus D $, if $v$ has a parent $u \in P(T)$, then the guard  on $u$ moves to $v$ and no guards move to $u$.
        \item \textbf{Shift to Parent:} For every $v \in V(\Phi) \cap D$, if $v$ has a parent $u \in P(T)$, then the guard on $v$ moves to $u$.
    \end{itemize}
\end{definition}
The next three lemmas form the key ingredients of the defense strategy used. They use the guard movements defined in Definition~\ref{def:cyclemovements} and Definition~\ref{def:spinemovements}.
 
\begin{lemma}\label{lem:fullspinemove}
     Let $\mathcal{C}=\{C_1,C_2\}$ be the \texttt{EVCC} returned by Algorithm~\ref{algo:main}. Let $\Phi = G[P(T)]$.
    \begin{itemize}
        \item     Starting from $C_1$, on applying  a hole-child exchange on $\Phi$ with no guard movements between the spine and the cycle, the resultant guard positions on the spine will be same as they are in the configuration $C_2$.
        \item     Starting from $C_2$, on applying a hole-parent exchange on $\Phi$ with no guard movements between the spine and the cycle, the resultant guard positions on the spine will be same as they are in the configuration $C_1$.
    \end{itemize}
\end{lemma}
\begin{figure}[t]
    \centering
\begin{minipage}{\textwidth}
    \centering
    \scriptsize
    \vspace{-0.5cm}
\begin{tikzpicture}[scale=0.9,
    label distance= -0.1cm,
    every node/.style={circle,fill=black,inner sep=2pt}]

\node[fill=darkgreen] (x) at (0,2) {};

\node[fill=darkgreen] (ci) at (8,2) {};

\node[fill=darkgreen] (a) at (13,2) {};
\node[] (y) at (0,0) {};
\node[label=below:$x_j$] (xj) at (8,1) {};
\node[fill=darkgreen,label=above:$x_i$] (xi) at (4,1) {};
\node[fill=darkgreen] (z) at (1,1) {};
\node[] (cj) at (4,0) {};
\node[fill=darkgreen ] (b) at (13,0) {};

\node[fill=darkgreen] (a1) at (2,1) {};
\node[] (a2) at (3,1) {};
\node[] (a3) at (5,1) {};
\node[fill=darkgreen] (a4) at (6,1) {};
\node[fill=darkgreen] (a5) at (7,1) {};
\node[fill=darkgreen] (a6) at (9,1) {};
\node[] (a7) at (10,1) {};
\node[fill=darkgreen] (a8) at (11,1) {};

\node[circle=none,draw=none,fill=none] at ($(z)+(0,-0.3)$) {\scriptsize \color{red} $\times$};

\node[] (c) at (12,1) {};

\draw[dashed] (x)--(y);

\draw (x)--(z);
\draw (y)--(z);



\draw (c)--(a);

\draw[dashed] (b)--(a);
\draw (c)--(b);

\draw[dashed] (x)--(ci);
\draw[dashed] (ci)--(a);
\draw[dashed] (b)--(cj);
\draw[dashed] (cj)--(y);

\draw (cj)--(xi);
\draw (z)--(a1);
\draw (a1)--(a2);
\draw (xi)--(a2);
\draw (a3)--(xi);
\draw (a3)--(a4);
\draw (a5)--(a4);
\draw (xj)--(a5);
\draw (xj)--(a6);
\draw (a7)--(a6);
\draw (a7)--(a8);
\draw (c)--(a8);
\draw (ci)--(xj);


\node[circle=none,draw=none,fill=none] at ($(a1)+(0.2, 0.5)$) {\tiny $\Phi_1 :$Hole-child exchange};

\node[circle=none,draw=none,fill=none] at ($(a4)+(0, 0.5)$) {\tiny $\Phi_2 :$Shift to child};

\node[circle=none,draw=none,fill=none] at ($(a7)+(0.5, 0.5)$) {\tiny $\Phi_3 :$Hole-child exchange};


\draw[dashed,blue]
($(z)+(-0.1,0.2)$) -- ($(z)+(-0.1,0.3)$) -- ($(a2)+(0.1,0.3)$) -- ($(a2)+(0.1,0.2)$) ; 

\draw[dashed,blue]
($(a3)+(-0.1,0.2)$) -- ($(a3)+(-0.1,0.3)$) -- ($(a5)+(0.1,0.3)$) -- ($(a5)+(0.1,0.2)$) ; 

\draw[dashed,blue]
($(a6)+(-0.1,0.2)$) -- ($(a6)+(-0.1,0.3)$) -- ($(c)+(0.1,0.3)$) -- ($(c)+(0.1,0.2)$) ; 

\draw[thick,dashed,red,-{Triangle}]
($(ci)+(0.2,-0.1)$) to[out=270,in=90] ($(ci)+(0.2,-0.8)$) ;

\draw[thick,dashed,red,-{Triangle}]
  ($(cj)+(0.2,0.8)$) to[out=270,in=90]  ($(cj)+(0.2,0.1)$);


\draw[thick,dashed,red,-{Triangle}]
($(a1)+(0.1,-0.25)$) -- ($(a2)+(-0.1,-0.25)$) ;

\draw[thick,dashed,red,-{Triangle}]
($(a4)+(-0.1,-0.25)$) -- ($(a3)+(0.1,-0.25)$) ;

\draw[thick,dashed,red,-{Triangle}]
($(a5)+(-0.1,-0.25)$) -- ($(a4)+(0.1,-0.25)$) ;

\draw[thick,dashed,red,-{Triangle}]
($(a6)+(0.1,-0.25)$) -- ($(a7)+(-0.1,-0.25)$) ;

\draw[thick,dashed,red,-{Triangle}]
($(a8)+(0.1,-0.25)$) -- ($(c)+(-0.1,-0.25)$) ;
  
\end{tikzpicture}
\begin{center}
\scriptsize
\vspace{-0.3cm}
{Configuration $C_1$}
\end{center}
\vspace{0.1cm}
\end{minipage}  

\begin{minipage}{\textwidth}
    \centering
    \scriptsize
\begin{tikzpicture}[scale=0.9,
    label distance= -0.1cm,
    every node/.style={circle,fill=black,inner sep=2pt}]

\node[fill=darkgreen] (x) at (0,2) {};

\node[] (ci) at (8,2) {};

\node[fill=darkgreen] (a) at (13,2) {};
\node[] (y) at (0,0) {};
\node[fill=darkgreen,label=below:$x_j$] (xj) at (8,1) {};
\node[label=above:$x_i$] (xi) at (4,1) {};
\node[fill=darkgreen] (z) at (1,1) {};
\node[fill=darkgreen] (cj) at (4,0) {};
\node[fill=darkgreen ] (b) at (13,0) {};

\node[] (a1) at (2,1) {};
\node[fill=darkgreen] (a2) at (3,1) {};
\node[fill=darkgreen] (a3) at (5,1) {};
\node[fill=darkgreen] (a4) at (6,1) {};
\node[] (a5) at (7,1) {};
\node[] (a6) at (9,1) {};
\node[fill=darkgreen] (a7) at (10,1) {};
\node[] (a8) at (11,1) {};

\node[fill=darkgreen] (c) at (12,1) {};

\draw[dashed] (x)--(y);

\draw (x)--(z);
\draw (y)--(z);



\draw (c)--(a);
\draw[dashed] (b)--(a);
\draw (c)--(b);

\draw[dashed] (x)--(ci);
\draw[dashed] (ci)--(a);
\draw[dashed] (b)--(cj);
\draw[dashed] (cj)--(y);

\draw (cj)--(xi);
\draw (z)--(a1);
\draw (a1)--(a2);
\draw (xi)--(a2);
\draw (a3)--(xi);
\draw (a3)--(a4);
\draw (a5)--(a4);
\draw (xj)--(a5);
\draw (xj)--(a6);
\draw (a7)--(a6);
\draw (a7)--(a8);
\draw (c)--(a8);
\draw (ci)--(xj);


\node[circle=none,draw=none,fill=none] at ($(a1)+(0.2, 0.5)$) {\tiny $\Phi_1$};

\node[circle=none,draw=none,fill=none] at ($(a4)+(0, 0.5)$) {\tiny $\Phi_2 $};

\node[circle=none,draw=none,fill=none] at ($(a7)+(0.5, 0.5)$) {\tiny $\Phi_3 $};


\draw[dashed,blue]
($(z)+(-0.1,0.2)$) -- ($(z)+(-0.1,0.3)$) -- ($(a2)+(0.1,0.3)$) -- ($(a2)+(0.1,0.2)$) ; 

\draw[dashed,blue]
($(a3)+(-0.1,0.2)$) -- ($(a3)+(-0.1,0.3)$) -- ($(a5)+(0.1,0.3)$) -- ($(a5)+(0.1,0.2)$) ; 

\draw[dashed,blue]
($(a6)+(-0.1,0.2)$) -- ($(a6)+(-0.1,0.3)$) -- ($(c)+(0.1,0.3)$) -- ($(c)+(0.1,0.2)$) ; 
\end{tikzpicture}
\begin{center}
\vspace{-0.3cm}
{Resultant configuration}
\end{center}
\end{minipage}  
    \caption{Guard movements on spine vertices while reconfiguring from $C_1$ as mentioned in Lemma~\ref{lem:splitspinemove}. A {\color{red}$\times$} symbol indicates that the guard on that vertex does not move. }
    \label{fig:splitspinemove}
\end{figure}

\begin{lemma}\label{lem:splitspinemove}
     Let $\mathcal{C}=\{C_1,C_2\}$ be the eternal vertex cover class returned by Algorithm~\ref{algo:main}. Let $x_i,x_j$ be two distinct vertices on the spine where $x_i$ is a descendant of $x_j$. Let $\Phi_1$, $\Phi_2$ and $\Phi_3$ be vertex disjoint subpaths of the spine as defined below:
     \begin{eqnarray*}
         \Phi_1& = &\begin{cases}
             \emptyset, ~~\text{if }x_i =f \\
             \text{subpath of the spine between }f \text{ and the child of }x_i, ~~\text{otherwise.} 
         \end{cases} \\
         \Phi_2& = &\begin{cases}
             \emptyset, ~~\text{if parent of }x_i \text{ is }x_j \\
             \text{subpath of the spine between  the parent of }x_i  \text{ and the child of } x_j,\\ \hspace{9cm} \text{otherwise.} 
         \end{cases} \\
         \Phi_3& = &\begin{cases}
             \emptyset, ~~\text{if }x_j =r \\
             \text{subpath of the spine between }r \text{ and the parent of }x_j, ~~\text{otherwise.} 
         \end{cases}
     \end{eqnarray*}
        \begin{itemize}
            \item Starting from $C_1$, with $x_j \notin C_1$ or $x_j =r$, and $x_i \in C_1$, suppose that in a movement (i) the guard from $x_i$ is moved to one of its cycle neighbours, (ii) a shift to child is applied on $\Phi_2$, (iii) a hole-child exchange is applied on $\Phi_1$ and $\Phi_3$, (iv) if $x_j \in C_1$ then the guard on $x_j$ moves to its child, and (v)  $x_j$ is getting a guard from one of its cycle neighbours. Then the resultant guard positions on the spine will be same as they are in the configuration $C_2$.
            \item Starting from $C_2$, with $x_i \notin C_2$ or $x_i =f$, and $x_j \in C_2$, suppose that in a movement (i) the guard from $x_j$ is moved to one of its cycle neighbours, (ii) a shift to parent is applied on $\Phi_2$, (iii) a hole-parent exchange is applied on $\Phi_1$ and $\Phi_3$, (iv) if $x_i \in C_2$ then the guard on $x_i$ moves to its parent, and (v)  $x_i$ is getting a guard from one of its cycle neighbours. Then the resultant guard positions on the spine will be same as they are in the configuration $C_1$.
        \end{itemize}
    
\end{lemma}
Refer to Figure~\ref{fig:splitspinemove} for an illustration of the movements mentioned in Lemma~\ref{lem:splitspinemove}.
\begin{lemma}\label{lem:spiltcyclemove}
    Let $\mathcal{C}=\{C_1,C_2\}$ be the eternal vertex cover class returned by Algorithm~\ref{algo:main}. Let $c_i$ and $c_j$ be two distinct vertices in $V(C)$. Let $\Psi_1$ and $\Psi_2$ be two vertex disjoint subpaths of $C$ as defined below.
    \begin{equation*}
        \Psi_1 = \begin{cases}
            \emptyset, &\text{if } i+1 =j \\
            c_{i+1},   &\text{if } i+1 =j-1 \\
            \text{anticlockwise path from }c_{i+1} \text{ to } c_{j-1}, &\text{otherwise.}
        \end{cases}
    \end{equation*}
    \begin{equation*}
        \Psi_2 = \begin{cases}
            \emptyset, &\text{if } j+1 =i \\
            c_{j+1},   &\text{if } j+1 =i-1 \\
            \text{anticlockwise path from }c_{j+1} \text{ to } c_{i-1}, &\text{otherwise.}
        \end{cases}
    \end{equation*}
    \begin{itemize}
        \item  Starting from $C_1$, with $c_j \in C_1$ and $c_i \notin C_1$, suppose that in a movement (i) $c_i$ is getting a guard from its spine neighbour, (ii) an anticlockwise shift is applied on $\Psi_1$, (iii) a restricted clockwise shift is applied on $\Psi_2$ and (iv) the guard from $c_j$ is moved to its spine neighbour. Then the resultant guard positions on $V(C)$ will be same as they are in the configuration $C_2$.
        \item     Starting from $C_2$, with $c_j \in C_2$ and $c_i \notin C_2$, suppose that in a movement (i) $c_i$ is getting a guard from its spine neighbour, (ii) a clockwise shift is applied on $\Psi_2$, (iii) a restricted anticlockwise shift is applied on $\Psi_1$ and (iv) the guard from $c_j$ is moved to its spine neighbour. Then the resultant guard positions on $V(C)$ will be same as they are in the configuration $C_1$.
    \end{itemize}
\end{lemma}
 \begin{figure}[t]
   \hspace{0.1cm}
    \begin{minipage}{0.9\textwidth}
    \scriptsize
    \centering
    \vspace{-1.5cm}
\begin{tikzpicture}[scale=0.86,
    label distance= -0.1cm,
    every node/.style={circle,fill=black,inner sep=2pt}]

\node[] (x) at (0,2) {};
\node[fill=darkgreen] (z) at (1,1) {};
\node[fill=darkgreen] (y) at (0,0) {};

\node[] (a) at (13,2) {};
\node[fill=darkgreen] (c) at (12,1) {};
\node[fill=darkgreen] (b) at (13,0) {};

\node[fill=darkgreen,label=below:$c_{i+1}$] (ci+1) at (7,2) {};
\node[label=above:$c_i$] (ci) at (8,2) {};
\node[fill=darkgreen,label=below:$c_{i-1}$] (ci-1) at (9,2) {};

\node[fill=darkgreen] (xj) at (8,1) {};
\node[] (xi) at (4,1) {};
\node[fill=darkgreen,label=above:$c_{j-1}$] (cj-1) at (3,0) {};
\node[fill=darkgreen,label=below:$c_j$] (cj) at (4,0) {};
\node[label=above:$c_{j+1}$] (cj+1) at (5,0) {};

\node[fill=darkgreen] (a1) at (2,2) {};
\node[] (a2) at (3,2) {};
\node[fill=darkgreen] (a3) at (4,2) {};
\node[fill=darkgreen] (a4) at (5,2) {};
\node[] (a5) at (6,2) {};
\node[] (a6) at (10,2) {};
\node[fill=darkgreen] (a7) at (11,2) {};
\node[] (a8) at (11,0) {};
\node[fill=darkgreen] (a9) at (10,0) {};
\node[fill=darkgreen] (a10) at (9,0) {};
\node[] (a11) at (8,0) {};
\node[fill=darkgreen] (a12) at (7,0) {};
\node[fill=darkgreen] (a13) at (6,0) {};
\node[] (a14) at (2,0) {};

\draw[] (x)--(z);
\draw[] (x)--(y);
\draw[] (y)--(z);

\draw[] (a)--(b);
\draw[] (a)--(c);
\draw[] (c)--(b);

\draw[] (x)--(a1);
\draw[] (a1)--(a2);
\draw[] (a3)--(a2);
\draw[] (a3)--(a4);
\draw[] (a5)--(a4);
\draw[] (a5)--(ci+1);
\draw[] (ci+1)--(ci);
\draw[] (ci-1)--(ci);
\draw[] (ci-1)--(a6);
\draw[] (a6)--(a7);
\draw[] (a)--(a7);
\draw[] (b)--(a8);
\draw[] (a9)--(a8);
\draw[] (a10)--(a9);
\draw[] (a11)--(a10);
\draw[] (a11)--(a12);
\draw[] (a13)--(a12);
\draw[] (a13)--(cj+1);
\draw[] (cj)--(cj+1);
\draw[] (cj-1)--(cj);
\draw[] (a14)--(cj-1);
\draw[] (a14)--(y);
\draw[dashed] (z)--(xi);
\draw[dashed] (xi)--(xj);
\draw[dashed] (c)--(xj);
\draw[] (xi)--(cj);
\draw[] (xj)--(ci);

\draw[thick,dashed,blue]
($(ci+1)+(0.2,0.1)$) -- ($(ci+1)+(0.2,0.3)$) -- ($(x)+(-0.3,0.3)$) -- ($(y)+(-0.3,-0.3)$) -- ($(cj-1)+(0.2,-0.3)$)--($(cj-1)+(0.2,-0.1)$); 

\draw[thick,dashed,blue]
($(ci-1)+(-0.2,0.1)$) -- ($(ci-1)+(-0.2,0.3)$) -- ($(a)+(0.3,0.3)$) -- ($(b)+(0.3,-0.3)$) -- ($(cj+1)+(-0.2,-0.3)$)--($(cj+1)+(-0.2,-0.1)$); 


\node[circle=none,draw=none,fill=none] at ($(a2)+(0.2, 0.55)$) {\scriptsize $\Psi_1 :$Anticlockwise shift};

\node[circle=none,draw=none,fill=none] at ($(a7)+(0.15, 0.55)$) {\scriptsize $\Psi_2 :$Restricted clockwise shift};

\draw[thick,dashed,red,-{Triangle}]
($(cj)+(0.2,0.2)$) -- ($(xi)+(0.2,-0.2)$) ;

\draw[thick,dashed,red,-{Triangle}]
($(xj)+(0.2,0.2)$) -- ($(ci)+(0.2,-0.2)$) ;

\draw[thick,dashed,red,-{Triangle}]
($(ci+1)+(-0.2,-0.2)$) -- ($(a5)+(0.2,-0.2)$) ;

\draw[thick,dashed,red,-{Triangle}]
($(a4)+(-0.2,-0.2)$) -- ($(a3)+(0.2,-0.2)$) ;

\draw[thick,dashed,red,-{Triangle}]
($(a3)+(-0.2,-0.2)$) -- ($(a2)+(0.2,-0.2)$) ;

\draw[thick,dashed,red,-{Triangle}]
($(a1)+(-0.2,-0.2)$) -- ($(x)+(0.5,-0.2)$) ;

\draw[thick,dashed,red,-{Triangle}]
($(y)+(0.4,0.2)$) -- ($(a14)+(-0.4,0.2)$) ;

\draw[thick,dashed,red,-{Triangle}]
($(cj-1)+(0.2,0.2)$) -- ($(cj)+(-0.2,0.2)$) ;

\draw[thick,dashed,red,-{Triangle}]
($(ci-1)+(0.2,-0.2)$) -- ($(a6)+(-0.2,-0.2)$) ;

\draw[thick,dashed,red,-{Triangle}]
($(a7)+(0.2,-0.2)$) -- ($(a)+(-0.5,-0.2)$) ;

\draw[thick,dashed,red,-{Triangle}]
($(b)+(-0.4,0.2)$) -- ($(a8)+(0.4,0.2)$) ;

\node[circle=none,draw=none,fill=none] at ($(a9)+(0,0.3)$) {\scriptsize \color{red} $\times$};

\draw[thick,dashed,red,-{Triangle}]
($(a10)+(-0.2,0.2)$) -- ($(a11)+(0.2,0.2)$) ;

\node[circle=none,draw=none,fill=none] at ($(a12)+(0,0.3)$) {\scriptsize \color{red} $\times$};

\draw[thick,dashed,red,-{Triangle}]
($(a13)+(-0.2,0.2)$) -- ($(cj+1)+(0.2,0.2)$) ;
 
\end{tikzpicture}
\vspace{-0.4cm}
\begin{center}
\scriptsize
{ Configuration $C_1$}
\end{center}
\vspace{-1pt}
    \end{minipage}
    
        \begin{minipage}{\textwidth}
    \scriptsize
    \centering
\begin{tikzpicture}[scale=0.9,
    label distance= -0.1cm,
    every node/.style={circle,fill=black,inner sep=2pt}]

\node[fill=darkgreen] (x) at (0,2) {};
\node[fill=darkgreen] (z) at (1,1) {};
\node[] (y) at (0,0) {};

\node[fill=darkgreen] (a) at (13,2) {};
\node[fill=darkgreen] (c) at (12,1) {};
\node[] (b) at (13,0) {};

\node[label={[yshift=3pt]below:$c_{i+1}$}] (ci+1) at (7,2) {};
\node[fill=darkgreen,label=above:$c_i$] (ci) at (8,2) {};
\node[label={[yshift=3pt]below:$c_{i-1}$}] (ci-1) at (9,2) {};

\node[] (xj) at (8,1) {};
\node[fill=darkgreen] (xi) at (4,1) {};
\node[label={[yshift=-3pt]above:$c_{j-1}$}] (cj-1) at (3,0) {};
\node[fill=darkgreen,label=below:$c_j$] (cj) at (4,0) {};
\node[fill=darkgreen,label={[yshift=-3pt]above:$c_{j+1}$}] (cj+1) at (5,0) {};

\node[] (a1) at (2,2) {};
\node[fill=darkgreen] (a2) at (3,2) {};
\node[fill=darkgreen] (a3) at (4,2) {};
\node[] (a4) at (5,2) {};
\node[fill=darkgreen] (a5) at (6,2) {};
\node[fill=darkgreen] (a6) at (10,2) {};
\node[] (a7) at (11,2) {};
\node[fill=darkgreen] (a8) at (11,0) {};
\node[fill=darkgreen] (a9) at (10,0) {};
\node[] (a10) at (9,0) {};
\node[fill=darkgreen] (a11) at (8,0) {};
\node[fill=darkgreen] (a12) at (7,0) {};
\node[] (a13) at (6,0) {};
\node[fill=darkgreen] (a14) at (2,0) {};

\draw[] (x)--(z);
\draw[] (x)--(y);
\draw[] (y)--(z);

\draw[] (a)--(b);
\draw[] (a)--(c);
\draw[] (c)--(b);

\draw[] (x)--(a1);
\draw[] (a1)--(a2);
\draw[] (a3)--(a2);
\draw[] (a3)--(a4);
\draw[] (a5)--(a4);
\draw[] (a5)--(ci+1);
\draw[] (ci+1)--(ci);
\draw[] (ci-1)--(ci);
\draw[] (ci-1)--(a6);
\draw[] (a6)--(a7);
\draw[] (a)--(a7);
\draw[] (b)--(a8);
\draw[] (a9)--(a8);
\draw[] (a10)--(a9);
\draw[] (a11)--(a10);
\draw[] (a11)--(a12);
\draw[] (a13)--(a12);
\draw[] (a13)--(cj+1);
\draw[] (cj)--(cj+1);
\draw[] (cj-1)--(cj);
\draw[] (a14)--(cj-1);
\draw[] (a14)--(y);
\draw[dashed] (z)--(xi);
\draw[dashed] (xi)--(xj);
\draw[dashed] (c)--(xj);
\draw[] (xi)--(cj);
\draw[] (xj)--(ci);
\end{tikzpicture}
\vspace{-0.4cm}
\begin{center}
{Resultant configuration}
\end{center}
    \end{minipage}
    \caption{Guard movements on the cycle while reconfiguring from $C_1$ as mentioned in Lemma~\ref{lem:spiltcyclemove}. A {\color{red}$\times$} symbol indicates that the guard on that vertex does not move.}
    \label{fig:spiltcyclemove}
\end{figure}

Refer to Figure~\ref{fig:spiltcyclemove} for an illustration of the guard movements on the cycle vertices from $C_1$ to $C_2$ as given in Lemma~\ref{lem:spiltcyclemove}.

 The next lemma bounds the size of the configurations defined in Algorithm~\ref{algo:main}.
\begin{lemma}\label{lem:size_of_c1}
    Let $C_1$ be the configuration as defined in Algorithm~\ref{algo:main}.  Each vertex of $C_1 \setminus S$ can be mapped to three distinct vertices from $C_1 \cap S$ such that no vertex from $S$ is mapped more than once. Consequently, $|C_1| \leq \frac{4}{3}|S|$.  
\end{lemma}

For the proof, each $v \in F$ is mapped to child of $v$ and a cycle neighbour each of $v$ and its child. If $f \notin S$, it is mapped to two of its cycle neighbours and a cycle neighbour of $r$. 
All the steps of Algorithm~\ref{algo:main} can be done in linear time, including the computation of minimum vertex cover~\cite{BorieParkerTovey2008}. 
Lemma~\ref{lem:caterpillarproof} shows that for every caterpillar Halin graphs, Algorithm~\ref{algo:main} returns an \texttt{EVCC}. Combining this with Lemma~\ref{lem:size_of_c1}, gives us the following theorem. 

\begin{theorem}\label{thm:main_caterpillar_bound}
    For every caterpillar Halin graph $G$, $\evc(G) \leq \frac{4}{3}\mvc(G)$. Furthermore, an \texttt{EVCC} of $G$ of size at most $\frac{4}{3}\mvc(G)$ can be computed in linear time.  
\end{theorem}

\section{Discussion}\label{sec:discussion}
We have shown that for all Halin graphs $G$, $\evc(G) \le \frac{3}{2}\mvc(G)
$ and for any caterpillar Halin graph $G$, $\evc(G) \le \frac{4}{3}\mvc(G)$. Moreover, for any Halin graph $G$, $\evc(G) \le \frac{4}{3}\mvc(G)$, except when $\lceil\frac{n}{2}\rceil \le \mvc(G) <\frac{9n}{16}$. 
In addition, Halin graphs with $\evc(G) \le \frac{2n}{3}$, will have $\evc(G) \le \frac{4}{3}\mvc(G)$, because $\mvc(G) \ge \lceil\frac{n}{2}\rceil$ by the Hamiltonicity of $G$. 
All of these suggest that $\evc(G) \le \frac{4}{3}\mvc(G)$ is likely to be true for all Halin graphs.
On the lower bound side, if there is an infinite subfamily of Halin graphs with $\evc(G) \ge \frac{2n}{3}$ and $\mvc(G) \le \frac{9n}{16}$, it would imply that $\rho(\mathcal{H}) \ge \frac{32}{27}$, which would improve our current lower bound. 

Another interesting direction is to obtain good bounds for the parameter $\rho$ for the family of treewidth three biconnected graphs, which is a superclass of Halin graphs.  More generally, the fact that $\rho$ is $2$ for biconnected series parallel graphs (i.e., graphs of treewidth at most two), while it has a smaller lower bound for Halin graphs suggests us to explore the effect of larger treewidth on the parameter. Settling the complexity status of the eternal vertex cover problem of Halin graphs and series parallel graphs are other major questions.

\newpage

\newpage
\appendix

\section{Appendix} \label{app:proofs}

\begin{figure*}[!ht]
\centering
\large
\resizebox{0.8\textwidth}{!}{%
\begin{circuitikz}
\draw [short] (2.5,10.75) -- (2.5,8.375);
\draw [short] (2.5,10.75) -- (3.75,9.5);
\draw [short] (2.5,8.25) -- (3.75,9.5);
\draw [short] (6.25,9.5) -- (5,9.5);
\draw [short] (6.25,8.25) -- (5,8.25);
\draw [short] (6.25,9.625) -- (6.25,8.375);
\draw [short] (5,9.5) -- (3.75,9.5);
\draw [short] (5,8.25) -- (2.5,8.25);
\draw [short] (5,9.625) -- (5,8.375);
\draw [short] (7.5,9.5) -- (6.25,9.5);
\draw [short] (7.5,8.25) -- (6.25,8.25);
\draw [short] (7.5,9.625) -- (7.5,8.375);
\draw [short] (10,9.625) -- (10,8.375);
\draw [short] (11.25,9.5) -- (10,9.5);
\draw [short] (11.25,8.25) -- (10,8.25);
\draw [short] (11.25,9.625) -- (11.25,8.375);
\draw [short] (13.75,10.75) -- (13.75,8.375);
\draw [short] (12.5,9.5) -- (13.75,8.25);
\draw [short] (13.75,10.75) -- (12.5,9.5);
\draw [short] (12.5,9.5) -- (11.25,9.5);
\draw [short] (13.75,8.25) -- (11.25,8.25);
\draw [dashed] (7.5,9.5) -- (10,9.5);
\draw [dashed] (7.5,8.25) -- (10,8.25);
\draw [short] (2.5,10.75) -- (13.625,10.75);
\node [ inner xsep=0.080cm, inner ysep=0.085cm, rounded corners=0.000cm] at (5,7.875) {$b_1$};
\node [ inner xsep=0.080cm, inner ysep=0.085cm, rounded corners=0.000cm] at (6.25,7.875) {$b_2$};
\node [ inner xsep=0.080cm, inner ysep=0.085cm, rounded corners=0.000cm] at (7.5,7.875) {$b_3$};
\node [ inner xsep=0.080cm, inner ysep=0.085cm, rounded corners=0.019cm] at (10,7.875) {$b_{2t-1}$};
\node [ inner xsep=0.080cm, inner ysep=0.085cm, rounded corners=0.020cm] at (11.25,7.875) {$b_{2t}$};
\node [ inner xsep=0.080cm, inner ysep=0.085cm, rounded corners=0.000cm] at (5,9.875) {$a_1$};
\node [ inner xsep=0.080cm, inner ysep=0.085cm, rounded corners=0.000cm] at (6.25,9.875) {$a_2$};
\node [ inner xsep=0.080cm, inner ysep=0.085cm, rounded corners=0.000cm] at (7.5,9.875) {$a_3$};
\node [ inner xsep=0.080cm, inner ysep=0.085cm, rounded corners=0.012cm] at (10,9.875) {$a_{2t-1}$};
\node [ inner xsep=0.080cm, inner ysep=0.085cm, rounded corners=0.020cm] at (11.25,9.875) {$a_{2t}$};
\draw [ color={rgb,255:red,9; green,160; blue,17}, draw opacity=1 , fill={rgb,255:red,9; green,160; blue,17}, fill opacity=1]  (3.75,9.5) circle (0.125cm);
\draw [ color={rgb,255:red,9; green,160; blue,17}, draw opacity=1 , fill={rgb,255:red,9; green,160; blue,17}, fill opacity=1] (2.5,8.25) circle (0.125cm);
\draw [ color={rgb,255:red,0; green,0; blue,0}, draw opacity=1 , fill={rgb,255:red,0; green,0; blue,0}, fill opacity=1] (2.5,10.75) circle (0.125cm); 
\draw [ fill={rgb,255:red,31; green,41; blue,55}, fill opacity=1] (6.25,9.5) circle (0.125cm);
\draw [ color={rgb,255:red,9; green,160; blue,17}, draw opacity=1 , fill={rgb,255:red,9; green,160; blue,17}, fill opacity=1] (6.25,8.25) circle (0.125cm);
\draw [ color={rgb,255:red,9; green,160; blue,17}, draw opacity=1 , fill={rgb,255:red,9; green,160; blue,17}, fill opacity=1] (5,9.5) circle (0.125cm);
\draw [ fill={rgb,255:red,31; green,41; blue,55}, fill opacity=1] (5,8.25) circle (0.125cm);
\draw [ color={rgb,255:red,9; green,160; blue,17}, draw opacity=1 , fill={rgb,255:red,9; green,160; blue,17}, fill opacity=1] (7.5,9.5) circle (0.125cm);
\draw [ fill={rgb,255:red,31; green,41; blue,55}, fill opacity=1] (7.5,8.25) circle (0.125cm);
\draw [ color={rgb,255:red,9; green,160; blue,17}, draw opacity=1 , fill={rgb,255:red,9; green,160; blue,17}, fill opacity=1] (10,9.5) circle (0.125cm);
\draw [ fill={rgb,255:red,31; green,41; blue,55}, fill opacity=1] (10,8.25) circle (0.125cm);
\draw [ fill={rgb,255:red,31; green,41; blue,55}, fill opacity=1] (11.25,9.5) circle (0.125cm);
\draw [ color={rgb,255:red,9; green,160; blue,17}, draw opacity=1 , fill={rgb,255:red,9; green,160; blue,17}, fill opacity=1] (11.25,8.25) circle (0.125cm);
\draw [ color={rgb,255:red,9; green,160; blue,17}, draw opacity=1 , fill={rgb,255:red,9; green,160; blue,17}, fill opacity=1] (13.75,10.75) circle (0.125cm);
\draw [ fill={rgb,255:red,31; green,41; blue,55}, fill opacity=1] (13.75,8.25) circle (0.125cm);
\draw [ color={rgb,255:red,9; green,160; blue,17}, draw opacity=1 , fill={rgb,255:red,9; green,160; blue,17}, fill opacity=1] (12.5,9.5) circle (0.125cm);
\end{circuitikz}
}%

\end{figure*}
\vspace{-0.5cm}
\begin{figure}[!ht]
\centering
\large
\resizebox{0.8\textwidth}{!}{%
\begin{circuitikz}
\draw [short] (2.5,10.75) -- (2.5,8.375);
\draw [short] (2.5,10.75) -- (3.75,9.5);
\draw [short] (2.5,8.25) -- (3.75,9.5);
\draw [short] (6.25,9.5) -- (5,9.5);
\draw [short] (6.25,8.25) -- (5,8.25);
\draw [short] (6.25,9.625) -- (6.25,8.375);
\draw [short] (5,9.5) -- (3.75,9.5);
\draw [short] (5,8.25) -- (2.5,8.25);
\draw [short] (5,9.625) -- (5,8.375);
\draw [short] (7.5,9.5) -- (6.25,9.5);
\draw [short] (7.5,8.25) -- (6.25,8.25);
\draw [short] (7.5,9.625) -- (7.5,8.375);
\draw [short] (10,9.625) -- (10,8.375);
\draw [short] (11.25,9.5) -- (10,9.5);
\draw [short] (11.25,8.25) -- (10,8.25);
\draw [short] (11.25,9.625) -- (11.25,8.375);
\draw [short] (13.75,10.75) -- (13.75,8.375);
\draw [short] (12.5,9.5) -- (13.75,8.25);
\draw [short] (13.75,10.75) -- (12.5,9.5);
\draw [short] (12.5,9.5) -- (11.25,9.5);
\draw [short] (13.75,8.25) -- (11.25,8.25);
\draw [dashed] (7.5,9.5) -- (10,9.5);
\draw [dashed] (7.5,8.25) -- (10,8.25);
\draw [short] (2.5,10.75) -- (13.625,10.75);
\node [ inner xsep=0.080cm, inner ysep=0.085cm, rounded corners=0.000cm] at (5,7.875) {$b_1$};
\node [ inner xsep=0.080cm, inner ysep=0.085cm, rounded corners=0.000cm] at (6.25,7.875) {$b_2$};
\node [ inner xsep=0.080cm, inner ysep=0.085cm, rounded corners=0.000cm] at (7.5,7.875) {$b_3$};
\node [ inner xsep=0.080cm, inner ysep=0.085cm, rounded corners=0.019cm] at (10,7.875) {$b_{2t-1}$};
\node [ inner xsep=0.080cm, inner ysep=0.085cm, rounded corners=0.020cm] at (11.25,7.875) {$b_{2t}$};
\node [ inner xsep=0.080cm, inner ysep=0.085cm, rounded corners=0.000cm] at (5,9.875) {$a_1$};
\node [ inner xsep=0.080cm, inner ysep=0.085cm, rounded corners=0.000cm] at (6.25,9.875) {$a_2$};
\node [ inner xsep=0.080cm, inner ysep=0.085cm, rounded corners=0.000cm] at (7.5,9.875) {$a_3$};
\node [ inner xsep=0.080cm, inner ysep=0.085cm, rounded corners=0.012cm] at (10,9.875) {$a_{2t-1}$};
\node [ inner xsep=0.080cm, inner ysep=0.085cm, rounded corners=0.020cm] at (11.25,9.875) {$a_{2t}$};
\draw [ color={rgb,255:red,9; green,160; blue,17}, draw opacity=1 , fill={rgb,255:red,9; green,160; blue,17}, fill opacity=1]  (3.75,9.5) circle (0.125cm);
\draw [ color=black, fill=black] (2.5,8.25) circle (0.125cm);
\draw [ color={rgb,255:red,9; green,160; blue,17}, draw opacity=1 , fill={rgb,255:red,9; green,160; blue,17}, fill opacity=1] (2.5,10.75) circle (0.125cm); 
\draw [ color={rgb,255:red,9; green,160; blue,17}, draw opacity=1 , fill={rgb,255:red,9; green,160; blue,17}, fill opacity=1] (6.25,9.5) circle (0.125cm);
\draw [ color=black, fill=black] (6.25,8.25) circle (0.125cm);
\draw [ color=black, fill=black] (5,9.5) circle (0.125cm);
\draw [ color={rgb,255:red,9; green,160; blue,17}, draw opacity=1 , fill={rgb,255:red,9; green,160; blue,17}, fill opacity=1] (5,8.25) circle (0.125cm);
\draw [ color=black, fill=black] (7.5,9.5) circle (0.125cm);
\draw [ color={rgb,255:red,9; green,160; blue,17}, draw opacity=1 , fill={rgb,255:red,9; green,160; blue,17}, fill opacity=1] (7.5,8.25) circle (0.125cm);
\draw [ color=black, fill=black] (10,9.5) circle (0.125cm);
\draw [ color={rgb,255:red,9; green,160; blue,17}, draw opacity=1 , fill={rgb,255:red,9; green,160; blue,17}, fill opacity=1] (10,8.25) circle (0.125cm);
\draw [color={rgb,255:red,9; green,160; blue,17}, draw opacity=1 , fill={rgb,255:red,9; green,160; blue,17}, fill opacity=1] (11.25,9.5) circle (0.125cm);
\draw [ color=black, fill=black] (11.25,8.25) circle (0.125cm);
\draw [ color=black, fill=black] (13.75,10.75) circle (0.125cm);
\draw [ color={rgb,255:red,9; green,160; blue,17}, draw opacity=1 , fill={rgb,255:red,9; green,160; blue,17}, fill opacity=1] (13.75,8.25) circle (0.125cm);
\draw [ color={rgb,255:red,9; green,160; blue,17}, draw opacity=1 , fill={rgb,255:red,9; green,160; blue,17}, fill opacity=1] (12.5,9.5) circle (0.125cm);
\node [ inner xsep=0.080cm, inner ysep=0.085cm, rounded corners=0.020cm] at (5.25,7.125) {Total number of vertices:  $4t +6$};
\node [ inner xsep=0.080cm, inner ysep=0.085cm, rounded corners=0.020cm] at (5,6.625) { By Algorithm 1, $k = 3t + 4$};
\node [ inner xsep=0.080cm, inner ysep=0.085cm, rounded corners=0.020cm] at (10.5,7.125) {$\mvc(G) = 2t +4$};
\node [ inner xsep=0.080cm, inner ysep=0.085cm, rounded corners=0.020cm] at (10.875,6.625) {Actual $\evc(G) = 2t + 4$};
\end{circuitikz}
}%
\caption{A graph $G$ for which Algorithm 1 requires close to $\frac{3}{2}\mvc(G)$ guards for defense. The two configurations shown are sufficient for defense using $2t+4$ guards.}
\label{fig:bad_example_algo1}
\end{figure}


\begin{proof}[Proof of  Proposition~\ref{prop:unit-mvc}]
Note that $G_r$ has exactly $6r+3$ vertices. Since $G_r$ is Hamiltonian, there exists a cycle in $G_r$ that contains all its vertices. Hence $\mvc(G_r) \ge 3r+2$. It is easy to construct a vertex cover of $G_r$ of size $3r+2$ by  keeping $3$ vertices in each unit.  $S=\{x,z\} \cup \{a_{i1},a_{i3},a_{i5}\,|\,1 \le i\le r \}$ gives the required vertex cover. Note that in any minimum vertex cover, there must be exactly $3$ vertices from each unit and $2$ vertices from the triangle $xyz$. To cover the edges of a unit $U_i$ with exactly three guards, the only possible way is to keep the guards on $a_{i1},a_{i3}$ and $a_{i5}$. This forces the guards on the triangle $xyz$ to be on $x$ and $z$.  Hence, $S$ is the unique minimum vertex cover of $G_r$. 
\end{proof} 
\begin{proof}[Proof of Lemma~\ref{lem:block_notlast}]
    Recall that each unit requires at least $3$ guards. Since, $B$ is not the last block, $B$ has at least one unit that has more than $3$ guards.  If $B$ has no units with exactly three guards, then the lemma holds trivially. Now if $B$ has exactly one three guarded unit, then all units except the first one has at least $4$ guards. Hence, the lemma holds. 
    
    In the remaining case, $B$ has exactly two units with exactly three guards each. Here we have $\eta(U_i) =\eta(U_{i+1}) =3$. If $B$ has $4$ or more units, then the lemma holds.  If $B$ has only 3 units then $U_{i+3}$ belongs to the next block and $\eta(U_{i+3}) = 3$. Now, by Lemma~\ref{lem:four_units}, it follows that $\eta(U_{i+2}) \ge 5$ and the lemma holds for the block $B$. 
\end{proof}

\begin{proof}[Proof of Thereom~\ref{thm:lower_bound_family}] Let $X$ be an eternal vertex cover class configuration of $G_r$.
 By Lemma~\ref{lem:block_notlast} and Lemma~\ref{lem:block_last}, the total number of guards on units $U_1$ to $U_r$ is at least $3.5r -1$ in $X$. The triangle $xyz$ requires at least two guards in $X$. From this, we get $\evc(G_r) \ge 3.5r +1$. By Proposition~\ref{prop:unit-mvc}, $\mvc(G_r)=3r+2 $. Hence $\liminf\limits_{r\rightarrow \infty} \frac{\evc(G_r)}{\mvc(G_r)} = \frac{7}{6} $. This shows that  $\rho(\mathcal{H}) \ge \frac{7}{6}$.
\end{proof}
\begin{proof}[Proof of Lemma~\ref{lem:nbound}]
 Consider the class $\mathcal{C}=\{C_1,C_2\}$ returned by the algorithm. Clearly, $|C_1|=|C_2|=\left\lceil\frac{l}{2}\right\rceil+|int(T)|=k$. We now show that $\mathcal{C}$ forms an eternal vertex cover class of $G$.
If an edge with both endpoints occupied in a configuration is attacked, then the attack can be defended easily by swapping  the guards on the endpoints of that edge and remain in the same configuration.
Hence, we need to only consider edges with at least one unoccupied endpoint. By our construction, the only unoccupied vertices are in the cycle. 
It is easy to see that the configurations $C_1,C_2$ given by Algorithm~\ref{algo:evchalin11} satisfy the condition in Observation~\ref{obs:fullcyclemove}.

Let $C_1$ \both{$C_2$} be the current configuration. Let $uv$ be the edge attacked where $v$ is the unoccupied vertex. By construction, $v=c_i \in V(C)$. 
Let $\Psi$ be a subpath of the cycle that contains all the vertices of $C$. If $u=c_{i-1}$ \both{$u=c_{i+1}$}, then an anticlockwise shift \both{clockwise shift} in $\Psi$ will defend the attack. 
Further, the resultant configuration will be $C_2$ \both{$C_1$}, by Observation~\ref{obs:fullcyclemove}. 
If $u=c_{i+1}$ \both{$u=c_{i-1}$}, then a restricted clockwise shift \both{restricted anticlockwise shift} in $\Psi$ will defend the attack and the resultant configuration will be $C_2$ \both{$C_1$}.
In the remaining case, $u \in int(T)$. Let $y=c_{i-1}$ \both{$y=c_{i+1}$} be the clockwise neighbour \both{anticlockwise neighbour} of $v$ on $C$. Let $x \in int(T)$ be a vertex such that $xy$ is an edge in $G$. Observe that since $v$ has no guard, $y$ will have a guard in the current configuration. Note that it is possible that $u=x$. Let $\Psi'$ be the subpath of $C$, that contains all the vertices of $C$ except $v$ and $y$. To defend the attack on edge $uv$, reconfigure in the following way. Along the unique tree path connecting $y$ and $v$, move all guards in the direction from $y$ to $v$.  On the cycle, do an anticlockwise shift \both{clockwise shift} on $\Psi'$. Now, we will show that the resultant configuration is $C_2$ \both{$C_1$}. First, observe that by the movement on the tree path, all the vertices in $int(T)$ remain occupied in the resultant configuration, as required. 
Further, observe that by the tree path movement, the guard from $y=c_{i-1}$ \both{$y=c_{i+1}$} is moved  and $v=c_{i}$ gets a guard. 
These movements, together with the anticlockwise shift \both{clockwise shift} performed on $\Psi'$, result in the same guard configuration on $C$ as that obtained by performing an anticlockwise shift \both{clockwise shift} on all vertices of the cycle.
Hence by Observation~\ref{obs:fullcyclemove}, the resultant configuration on the cycle is same as $C_2$ \both{$C_1$}.

Hence, every attack can be defended while shifting between the configurations of $\mathcal{C}$. Therefore, $\mathcal{C}$ is an eternal vertex cover class of $G$ and so  $\evc(G) \leq k$.

 It remains to show that $k \leq \frac{3n}{4}$. Note that $l$, which is the number of vertices in $C$, is equal to the number of leaves of $T$.
    Since the degree of every vertex in $int(T)$ is at least $3$, a simple combinatorial argument gives
       $ 1 \leq |int(T)| \leq l-2$
    and hence substituting $|int(T)|= n-l $ and rearranging, we get  $ l -1 \ge \frac{n}{2} $. Now, dividing both sides by two we get, $\frac{l}{2} - \frac{1}{2}  \geq \frac{n}{4} $. Now, it is easy to verify that $\left\lfloor \frac{l}{2} \right\rfloor \geq \frac{l}{2} - \frac{1}{2}$. This gives us $\left\lfloor \frac{l}{2} \right\rfloor \geq \frac{n}{4}$. 
Since $k = \left\lceil \frac{l}{2} \right\rceil + |int(T)|$ and $|int(T)|= n -l$, we have
 $  k = n - \left\lfloor \frac{l}{2} \right\rfloor $.
Using the fact that $\left\lfloor \frac{l}{2} \right\rfloor \geq \frac{n}{4}$, we get 
   $ k \le n -  \frac{n}{4} \le  \frac{3n}{4}$.
\end{proof}
\begin{proof}[Proof of Lemma~\ref{lem:neighbour_of_child_of_hole_in_C1}]
Consider any $x \in P(T) \setminus C_1$. Let the child of $x$ be referred to as $c(x)$. If $x \neq r$, then by the definition of the set $F$ in Algorithm~\ref{algo:main}, $c(x)$ has an unoccupied neighbour on $C$. If $x=r$, then since $r=x \notin C_1$, all its neighbours are in $C_1$. 
Further, the cycle neighbours of  $r$ and one of the cycle neighbours of $c(x)$ will be consecutive on the cycle. 
If this cycle neighbour of $c(x)$ is in $C_1$, it will contradict our choice of $C_1$ because in $C_1$ there are no three consecutive vertices from the cycle. Therefore, part $(i)$ of the lemma holds true.

Now consider any $x \in P(T) \setminus C_2$. Note that no three consecutive vertices from the cycle are in $C_2$. This is due to our choice of $C_1$ and the construction of $C_2$. By Observation~\ref{obs:f_in_c1}, $r\in C_2$ and hence $x \neq r$. If any of the cycle neighbours of $r$ are not occupied in $C_2$, then $r$ is an ancestor of $x$ as required by the lemma. 
Henceforth, assume that all the cycle neighbours of $r$ are occupied in $C_2$. 
Let $c(r)$ be the child of $r$ in the spine. The cycle neighbours of $r$ and one of the cycle neighbours of $c(r)$, say $c_t$, are consecutive on the cycle. Since $r$ has at least two cycle neighbours and no three consecutive vertices on the cycle are occupied in $C_2$, $c_t \notin C_2$.
Hence, $c(r) \in C_2$ and $x \neq c(r)$. Furthermore, for every $x \in P(T)\setminus C_2$, choosing $x' = c(r)$ suffices because $c_t \notin C_2$ is a cycle neighbour of $c(r)$.
 Therefore, part $(ii)$ of the lemma holds true.
\end{proof}

\begin{proof}[Proof of Lemma~\ref{lem:caterpillarproof}]
First, we show that the vertex sets $C_1$ and $C_2$ defined in Algorithm~\ref{algo:main} are vertex covers of $G$ and $|C_1|=|C_2|$.

Since $C_1 \supseteq S$, $C_1$ is a vertex cover of $G$.
In $C_2$,  it is easy to see that all the edges of the cycle are covered.
    Now consider an edge $ab$ on the spine. For contradiction, assume that $a$ and $b$ were not in $C_2$. By Observation~\ref{obs:f_in_c1}, $a,b  \neq r $. Therefore,  $a$ and $b$ are not in $P_2$ and therefore their parents, which are consecutive vertices in the spine, are not in $C_1$. This leads to a contradiction as $C_1$ is a vertex cover of $G$. Hence, all edges of the spine are covered in $C_2$. 
    Now consider edges having one endpoint on the spine ($a$) and one on the cycle ($b = c_i$). Suppose that $a$ is not occupied in $C_2$. This means $a \notin P_2$. Now, the parent of $a$, say $p(a)$ is not in $C_1$. By the definition of $C_1$ and $F$, $a$ has  no neighbours in $S \cap V(C)$. Since $b$ is in $V(C)$, this implies $b = c_i \notin S$. Since $C_1$ is a vertex cover, $c_{i-1} \in C_1\cap V(C)$. By the definition of $Q_2$,  $b$ will be in $C_2$. Hence, the edge $ab$ is covered in $C_2$. This shows that $C_2$ is a valid vertex cover of $G$. 

  Now, we show that $|C_1|=|C_2|$. By definition of $Q_2$,  $|C_1 \cap V(C)|= |C_2 \cap V(C)|$. Hence, it suffices to show that $|C_1 \cap P(T)|= |C_2 \cap P(T)|$. For this, we define a bijection $g$ between $C_1 \cap P(T)$ and $C_2 \cap P(T)$ as follows.
  \[
  g(x) = \begin{cases}
      r,  &\text{if } x=f \\
      \text{child of }x, &\text{otherwise}
  \end{cases}
  \]
This mapping works for $x=f$ by Observation~\ref{obs:f_in_c1} and for other vertices by the definition of $P_2$. Thus, $|C_1 \cap P(T)| = |C_2 \cap P(T)|$ which implies $|C_1| = |C_2|$.

    It remains to prove that a movement strategy for reconfiguration between these configurations exists in order to defend an attack on any edge.  We prove the result simultaneously for both configurations. Let the attacked edge be $ab$. Without loss of generality, assume that either $a$ or $b$ is not occupied in the current configuration.

    \textbf{Case 1: Both endpoints are in $V(C)$ :}
    Without loss of generality, let $a= c_i$ and $b=c_{i+1}$ for some $1\le i \le |V(C)|$. The attack can be defended by simultaneously performing the guard movements defined in Observation~\ref{obs:fullcyclemove}  and Lemma~\ref{lem:fullspinemove}. Suppose the current configuration is $C_1$. If $c_{i+1} \notin C_1$ then do the anticlockwise shift on the cycle and hole-child exchange on the entire spine. If $c_{i} \notin C_1$ then do the restricted clockwise shift on the cycle and hole-child exchange on the entire spine. By Observation~\ref{obs:fullcyclemove}  and Lemma~\ref{lem:fullspinemove} the resultant configuration is $C_2$. It is easy to see that the attack is defended. 

    If the current configuration is $C_2$, we can apply symmetric arguments using the guard movements defined in Observation~\ref{obs:fullcyclemove}  and Lemma~\ref{lem:fullspinemove}. 
    
    \textbf{Case 2: $a$ is in spine and $b$ is on the cycle :} 
     First, suppose $b$ is unoccupied and current configuration is $C_1$ \both{$C_2$}. To defend the attack, the following guard movements are simultaneously applied (i) the guard movements as defined in Lemma~\ref{lem:spiltcyclemove} with $c_i = b$ and $c_j$ being an occupied neighbour of $r$ \both{$f$} and (ii) the guard movements as defined in Lemma~\ref{lem:splitspinemove} with $x_i=a$ and $x_j=r$ \both{$x_i=f$ and $x_j=a$}. While doing this the guard on $a$ moves to $b$ and the guard on $c_j$ moves to $r$ \both{$c_j$ moves to $f$}. Using Lemma~\ref{lem:splitspinemove} and Lemma~\ref{lem:spiltcyclemove}, we can see that the resultant configuration is $C_2$ \both{$C_1$}.

     Now, suppose $a$ is unoccupied and the current configuration is $C_1$ \both{$C_2$}.
     To defend the attack the following guard movements are simultaneously applied  (i) the guard movements as defined in Lemma~\ref{lem:spiltcyclemove} with $c_j = b$ and $c_i$ being an unoccupied neighbour of the child of $a$ \both{unoccupied neighbour of the ancestor $a'$ of $a$} given by Lemma~\ref{lem:neighbour_of_child_of_hole_in_C1} and (ii) the guard movements as defined in Lemma~\ref{lem:splitspinemove}   with $x_j=a$ and $x_i$ being child of $a$ \both{ $x_j=a'$  and $x_i=a$}. While doing this, the guard on $b$ moves to $a$ and the guard on child of $a$ moves to $c_i$ \both{$a'$ moves to $c_i$}. The feasibility of this strategy  is justified by Lemma~\ref{lem:neighbour_of_child_of_hole_in_C1}. Now, using Lemma~\ref{lem:splitspinemove} and Lemma~\ref{lem:spiltcyclemove}, we can see that the resultant configuration is $C_2$ \both{$C_1$}.

    \textbf{Case 3: Both endpoints are on the spine. } Without loss of generality, let $a$ be the child of $b$ and the current configuration be $C_1$ \both{$a$ be the parent of $b$ and current configuration be $C_2$}.
    
    First, suppose $b$ is unoccupied. In this case, the attack can be defended by applying hole-child exchange \both{hole-parent exchange} on the entire spine and anti-clockwise shift \both{clockwise shift} on a subpath of the cycle containing all its vertices. By Observation~\ref{obs:fullcyclemove}  and Lemma~\ref{lem:fullspinemove} the resultant configuration is $C_2$ \both{$C_1$} and the attack is defended. 
    
    Now, suppose $a$ is unoccupied. Let $c(a)$ be the child of $a$ on spine \both{$a'$ be an ancestor of $a$ given by Lemma~\ref{lem:neighbour_of_child_of_hole_in_C1}}. To defend the attack the following guard movements are applied simultaneously (i) the guard movements on the cycle as defined in Lemma~\ref{lem:spiltcyclemove} with $c_i$ as an unoccupied cycle neighbour of $c(a)$ \both{$c_i$ be an unoccupied neighbour of $a'$ in $C_2$} and $c_j$ being an occupied cycle neighbour of $r$ \both{$f$} and (ii) the guard movements on the spine as defined in Lemma~\ref{lem:splitspinemove} with $x_i=c(a)$ and $x_j =r$ \both{$x_i=f$ and $x_j=a'$}. While doing this, the guard on $b$ moves to $a$, the guard on $c(a)$ moves to $c_i$ \both{guard on $a'$ moves to $c_i$} and the guard on $c_j$ moves to $r$ \both{guard on $c_j$ moves to $f$}. The feasibility of this strategy is justified by Lemma~\ref{lem:neighbour_of_child_of_hole_in_C1}. By Lemma~\ref{lem:spiltcyclemove} and Lemma~\ref{lem:splitspinemove}, we can see that the resultant configuration is $C_2$ \both{$C_1$}.    
    
    This shows that $\mathcal{C}$ is an Eternal Vertex Cover Class of $G$.
    \end{proof}

\begin{proof}[Proof of Lemma~\ref{lem:fullspinemove}]
        Let $C_1$ be the starting configuration. Suppose, we perform hole-child exchange on the whole spine as mentioned in the first part of the lemma. It may be observed that no guard on the spine is moved more than once. Further, in hole-child exchange movement, a vertex $x_t$ can get a guard only from its child. We will show that a vertex $x_t$ of the spine is in the resultant configuration if and only if it is in $ C_2$. We will also show that in the resultant configuration, no vertex has multiple guards.

Since $r \in C_2$, we first show that $r$ is in the resultant configuration with exactly one guard. 
If $r \in C_1$, then since $r$ has no parent, by hole-child exchange, the guard on $r$ will not move. Also no other guard will move to $r$. Hence, as required, $r$ will be occupied with exactly one guard in the resultant configuration.
If $r \notin C_1$, then since $C_1$ is a vertex cover, child of $r$ on the spine, say $c(r)$, will be occupied in $C_1$.
Now, by hole-child exchange movement on $\Phi$, the guard on $c(r)$ will move to $r$. Also, no other guard moves to $r$.  Thus, in all cases, $r$ will be occupied with exactly one guard in the resultant configuration.

Henceforth, assume that $x_t \neq r$.  By Algorithm~\ref{algo:main}, $x_{t} \in C_2$ if and only if the parent of $x_t$ is in $C_1$. Let $p(x_t)$ denote the parent of $x_t$. We will show that $x_t$ is occupied in the resultant configuration if and only if $p(x_t) \in C_1$.
 Since $f$ has no child, $p(x_t) \neq f$.
 \begin{itemize}
     \item  \textbf{Case 1:} $p(x_t) \in C_1$. 
     We need to show that $x_t$ is occupied with exactly one guard in the resultant configuration.
     
First, consider the case when $x_t \in C_1$. On performing the hole-child exchange movement on $\Phi$, the guard on $x_t$ will not move to $p(x_t)$ as $p(x_t) \in C_1$. Also no other guard will move to $x_t$ as $x_t \in C_1$. Hence, as required, in the resultant configuration, $x_t$ is occupied with exactly one guard. 

If $x_t \notin C_1$, then by construction, $x_t \neq f$ and therefore, $x_t$ has a child, say $c(x_t)$, on the spine. Further, since $C_1$ is a vertex cover, $c(x_t) \in C_1$. Now, on performing the hole-child exchange on $\Phi$, the guard on $c(x_t)$ will move to $x_t$. Also, no other guard moves to $x_t$.  Hence, as required, in the resultant configuration, $x_t$ is occupied with exactly one guard. 

\item \textbf{Case 2:} $p(x_t) \notin C_1$. We need to show that $x_t$ has no guard in the resultant configuration. 

Since $C_1$ is a vertex cover, $x_t \in C_1$. On performing hole-child exchange on $\Phi$, the guard on $x_t$ will move to $p(x_t)$. Now, since $x_t \in C_1$, no guard will come on $x_t$ from its child. Hence, as required, $x_t$ is not occupied in the resultant configuration.
 \end{itemize}
Hence, in all cases, the resultant guard positions on the spine after the movements as per the lemma are the same as those in $C_2$ given by Algorithm~\ref{algo:main}. This completes the proof of the first part of the lemma.

For the second part of the lemma, the proof is symmetric with the following role exchanges: 
\begin{itemize}
    \item $C_1$ with $C_2$
    \item $r$ with $f$
    \item  hole-child exchange movement with hole-parent exchange movement 
    \item parent with child 
    \item $c(v)$ with $p(v)$ for any vertex $v$.
\end{itemize}

For the sake of completeness, we now give the complete proof for the second part of the lemma. The reader may choose to skip this section of the proof.

Let $C_2$ be the starting configuration. Suppose, we perform hole-parent exchange on the whole spine as mentioned in the second part of the lemma. It may be observed that no guard is moved more than once. Further, in hole-parent exchange movement, a vertex $x_t$ can get a guard only from its parent. We will show that a vertex $x_t$ of the spine is in the resultant configuration if and only if it is in $ C_1$. We will also show that in the resultant configuration, no vertex has multiple guards.

Since $f \in C_1$, we first show that $f$ is in the resultant configuration with exactly one guard. 
If $f \in C_2$, then since $f$ has no child, by hole-parent exchange, the guard on $f$ will not move. Also, no other guard will move to $f$. Hence, as required, $f$ will be occupied with exactly one guard in the resultant configuration.
If $f \notin C_2$, then since $C_2$ is a vertex cover, parent of $f$ on the spine, say $p(f)$, will be occupied in $C_2$.
Now, by hole-parent exchange movement on $\Phi$, the guard on $p(f)$ will move to $f$. Also, no other guard moves to $f$.  Thus, in all cases, $f$ will be occupied with exactly one guard in the resultant configuration.

Henceforth, assume that $x_t \neq f$.  By Algorithm~\ref{algo:main}, $x_{t} \in C_1$ if and only if the child of $x_t$ is in $C_2$. Let $c(x_t)$ denote the child of $x_t$. We will show that $x_t$ is occupied in the resultant configuration if and only if $c(x_t) \in C_2$.
 Since, $r$ has no parent, $c(x_t) \neq r$.
 \begin{itemize}
     \item  \textbf{Case 1:} $c(x_t) \in C_2$. 
     We need to show that $x_t$ is occupied with exactly one guard in the resultant configuration.
     
First, consider the case when $x_t \in C_2$. On performing the hole-parent exchange movement on $\Phi$, the guard on $x_t$ will not move to $c(x_t)$, as $c(x_t) \in C_2$. Also no other guard will move to $x_t$ as $ x_t \in C_2$. Hence, as required in the resultant configuration, $x_t$ is occupied with exactly one guard. 

If $x_t \notin C_2$, then by construction, $x_t \neq r$ and therefore, $x_t$ has a parent, say $p(x_t)$, on the spine. Further, since $C_2$ is a vertex cover, $p(x_t) \in C_2$. Now, on performing the hole-parent exchange on $\Phi$, the guard on $p(x_t)$ will move to $x_t$. Also, no other guard moves to $x_t$.  Hence, as required in the resultant configuration, $x_t$ is occupied with exactly one guard. 

\item \textbf{Case 2:} $c(x_t) \notin C_2$. We need to show that $x_t$ has no guard in the resultant configuration. 

Since $C_2$ is a vertex cover, $x_t \in C_2$. On performing hole-parent exchange on $\Phi$, the guard on $x_t$ will move to $c(x_t)$. Now, since $x_t \in C_2$, no guard will come on $x_t$ from its parent. Hence, as required, $x_t$ is not occupied in the resultant configuration.

 \end{itemize}

Hence, in all cases, the resultant guard positions on the spine after the movements as per the lemma are the same as those in $C_1$ given by Algorithm~\ref{algo:main}.

\end{proof}
\begin{proof}[Proof of Lemma~\ref{lem:splitspinemove}]
Let $C_1$ be the starting configuration. Suppose we perform the movements as mentioned in the first part of the lemma. It may be observed that no guard is moved more than once. We will show that a vertex $x_t$ of the spine is in the resultant configuration if and only if it is in $ C_2$. We will also show that in the resultant configuration, no vertex has multiple guards.

Since $r \in C_2$, we first show that $r$ is in the resultant configuration with exactly one guard. Observe that either $r \in V(\Phi_3)$ or $r = x_j$. If $r \in V(\Phi_3)$ and $r \notin C_1$ then $r$ gets a guard from its child, by hole-child exchange. If $r \in V(\Phi_3)$ and $r \in C_1$, then by hole-child exchange, the guard on it will not move and no other guard comes to $r$.
If $r=x_j$ and $r \in C_1$ then, by the movement, the guard on $r$  moves to its child  and $r$ gets exactly one guard from $c_j$, an occupied cycle neighbour of $r$ in $C_1$.
If $r=x_j$ and $r \notin C_1$, then $r$ will have exactly one guard, the guard that moved from $c_j$. 
Thus, in all cases, $r$ is occupied in the resultant configuration and has exactly one guard. 

Henceforth, assume that $x_t \neq r$.  By Algorithm~\ref{algo:main}, $x_{t} \in C_2$ if and only if the parent of $x_t$ is in $C_1$. Let $p(x_t)$ denote the parent of $x_t$ and if $x_t \neq f$, let $c(x_t)$ denote the child of $x_t$ on the spine. We will show that $x_t$ is occupied in the resultant configuration if and only if $p(x_t) \in C_1$.
 Since $f$ has no child, $p(x_t) \neq f$.
 \begin{itemize}  
 \item \textbf{Case 1 :} $p(x_t) \in V(\Phi_1)$.
  In this case, as $p(x_t) \neq f$, $x_t \in V(\Phi_1)$.

  If $p(x_t) \notin C_1$, then $x_t$ would exist in $\Phi_1$ and $x_t \in C_1$. 
By performing hole-child exchange, guard on $x_t$ would move to $p(x_t)$ and no guard moves to $x_t$. So, as required, $x_t$ will not be occupied in the resultant configuration.

 If $p(x_t) \in C_1$, we need to show that $x_t$ is occupied with exactly one guard in the resultant configuration. 
 If $x_t \in C_1$, then the by hole-child exchange, the guard on $x_t$ will not move, as $p(x_t) \in C_1$ and no other guard will come on $x_t$. Hence, $x_t$ is occupied with exactly one guard in the resultant configuration as required.
 If $x_t \notin C_1$, then by Observation~\ref{obs:f_in_c1}, $x_t \neq f$. Hence, the child of $x_t$,  $c(x_t)$, would exist in $\Phi_1$. 
Further, since $C_1$ is a vertex cover by Lemma~\ref{lem:caterpillarproof}, $c(x_t) \in C_1$.  By hole-child exchange, the only guard that moves to $x_t$ is from $c(x_t)$. So $x_t$ will be occupied with exactly one guard in the resultant configuration as required. 

\item \textbf{Case 2 :} $p(x_t) =x_i$. Note that $x_i \in C_1$ by our choice of $x_i$. Hence, showing $x_t$ is occupied with exactly one guard in the resultant configuration will suffice.  
We have $x_t \in V(\Phi_1)$ as $p(x_t) \neq f$. If $x_t \notin C_1$, then as discussed in Case $1$, $x_t \neq f$ and its child, $c(x_t) \in V(\Phi_1) \cap C_1$. 
By hole-child exchange, the only guard that moves to $x_t$  is from $c(x_t)$. If $x_t \in C_1$, then by hole-child exchange, the guard on $x_t$ will not move because $p(x_t) \in C_1$. 
No other guard will come to $x_t$. Thus, as required, $x_t$ is occupied and has exactly one guard in the resultant configuration. 

\item \textbf{Case 3 :} $p(x_t) \in V(\Phi_2)$. 

If $p(x_t) \in C_1$, then by shift to child movement, the guard from $p(x_t)$ will move to $x_t$ in the resultant configuration. Further, no other guard moves to $x_t$. Note that this holds even when $x_t =x_i$.
If $x_t \in C_1 \cap V(\Phi_2)$, the guard on $x_t$ moves to its child, whereas if $x_t = x_i$, the guard on $x_t$ moves to one of its cycle neighbours.
Thus, $x_t$ is occupied with exactly one guard in the resultant configuration.

If $p(x_t) \notin C_1$, then $x_t \in C_1$. By the guard movements, no guard will move to $x_t$. Further, if $x_t$ is $x_i$, the guard on $x_t$ will move to one of its cycle neighbours. Otherwise, the guard on $x_t$ will move to $c(x_t)$ by the shift to child movement. Thus, as required, in the resultant configuration $x_t$ will have no guards. 

\item \textbf{Case 4 :} $p(x_t) =x_j$.

If $x_j\in C_1$, the guard from $x_j$ moves to its child $x_t$ during the movement. Now, if $x_t =x_i$, by guards movements, the guard from $x_i$ moves to  one of its cycle neighbours and no other guard moves to $x_i$ other than the guard from $p(x_t)=x_j$. If $x_t \neq x_i$, then $x_t \in \Phi_2$. By shift to child movement, if $x_t \in C_1$, then the guard on $x_t$ moves to $c(x_t)$ and the only guard that comes to $x_t$ is from $p(x_t)=x_j$. If $x_t \notin C_1$, then also no other guard comes to $x_t$ other than the guard from $p(x_t)$. Thus, as needed, $x_t$ has exactly one guard in the resultant configuration.

If $p(x_t)=x_j \notin C_1$, then we need to show that $x_t$ is not occupied in the resultant  configuration.
Since $C_1$ is a vertex cover, $x_t \in C_1$. 
If $x_t = x_i$, by the guard movements, the guard from $x_i$ moves to one of its cycle neighbours and no guard moves to $x_i$. If $x_t \neq x_i$, then $x_t \in \Phi_2$. By shift to child movement,  the guard on $x_t$ moves to $c(x_t)$ and no other guard comes to $x_t$. Thus, as needed, $x_t$ is not occupied in the resultant configuration.

Note that these arguments go through even when $x_j =r$.

\item \textbf{Case 5 :} $p(x_t) \in V(\Phi_3)$. Recall that by definition of $\Phi_3$, $\Phi_3 = \emptyset$ whenever $x_j =r$. Hence, $x_j \neq r$. By our choice of $x_j$, $x_j \notin C_1$. Now, there are two possibilities: either $x_t = x_j$ or $x_t \in V(\Phi_3)$.

If $x_t = x_{j}$ then $p(x_j) \in C_1$.
In this case, we want to show that $x_j$ has exactly one guard in the resultant configuration.
Now, by the guard movements defined in the lemma, $x_j$ gets a guard from one of its cycle neighbours.
Since, hole-child exchange is perform on $\Phi_3$, there is no guard movement from $p(x_j)$ to $x_j$. Also, no guard can move from $c(x_j)$ to $x_j$.
Hence $x_j$ is occupied by exactly one guard in the resultant configuration. 

Now, consider the case when $x_t \in \Phi_3$. If $p(x_t) \in C_1$ and $x_t \in C_1$ then hole-child exchange movement, no guard moves to $x_t$ and the guard on $x_t$ remains on $x_t$.
If $p(x_t) \in C_1$ and $x_t \notin C_1$, then $c(x_t) \in C_1$. This implies $c(x_t) \neq x_j$. Now, by hole-child exchange movement, $x_t$ gets a guard from its child and no other guard comes to $x_t$.
Hence, whenever $p(x_t) \in C_1$, $x_t$ gets exactly one guard in the resultant configuration after the movements.  If $p(x_t) \notin C_1$, then $x_t \in C_1$. By the hole-child exchange movement, the guard from $x_t$ moves to $p(x_t)$ and no other guard comes to $x_t$. Hence, as required, $x_t$ has no guard in the resultant configuration.

Note that these arguments go through even when $p(x_t)=r$.

 \end{itemize}
Hence, in all cases, the resultant guard positions on the spine are the same as those in $C_2$ given by Algorithm~\ref{algo:main}. 
This completes the proof of the first part of the lemma.

For the second part of the lemma, the proof is symmetric with the following role exchanges: 
\begin{itemize}
    \item $C_1$ with $C_2$
    \item $r$ with $f$
    \item $\Phi_1$ with $\Phi_3$
    \item $x_i$ with $x_j$
    \item  hole-child exchange movement with hole-parent exchange movement 
    \item shift to child with shift to parent
    \item parent with child 
    \item $c(v)$ with $p(v)$ for any vertex $v$.
\end{itemize}

For the sake of completeness, we now give the complete proof for the second part of the lemma. The reader may choose to skip this section of the proof.

Let $C_2$ be the starting configuration. Suppose we perform the movements as mentioned in the second part of the lemma. It may be observed that no guard is moved more than once. We will show that a vertex $x_t$ of the spine is in the resultant configuration if and only if it is in $ C_1$. We will also show that in the resultant configuration, no vertex has multiple guards.

Since $f \in C_1$, we first show that $f$ is in the resultant configuration with exactly one guard. Observe that either $f \in V(\Phi_1)$ or $f = x_i$. If $f \in V(\Phi_1)$ and $f \notin C_2$ then $f$ gets a guard from its parent, by hole-parent exchange. If $f \in V(\Phi_1)$ and $f \in C_2$, then by hole-parent exchange, the guard on it will not move and no other guard comes to $f$.
If $f=x_i$ and $f \in C_2$ then, by the movement, the guard on $f$  moves to its parent  and $f$ gets exactly one guard from $c_i$, an occupied cycle neighbour of $f$ in $C_2$.
If $f=x_i$ and $f \notin C_2$, then $f$ will have exactly one guard, the guard that moved from $c_i$. 
Thus, in all cases, $f$ is occupied in the resultant configuration and has exactly one guard. 

Henceforth, assume that $x_t \neq f$.  By Algorithm~\ref{algo:main}, $x_{t} \in C_1$ if and only if the child of $x_t$ in the spine is in $C_2$. Let $c(x_t)$ denote the child of $x_t$ in the spine and if $x_t \neq r$, let $p(x_t)$ denote the parent of $x_t$. We will show that $x_t$ is occupied in the resultant configuration if and only if $c(x_t) \in C_2$.
 Since $r$ has no parent, $c(x_t) \neq r$.
 \begin{itemize}  
 \item \textbf{Case 1 :} $c(x_t) \in V(\Phi_3)$.
  In this case, as $c(x_t) \neq r$, $x_t \in V(\Phi_3)$.

  If $c(x_t) \notin C_2$, then $x_t$ would exist in $\Phi_3$ and $x_t \in C_2$. 
By performing hole-parent exchange, guard on $x_t$ would move to $c(x_t)$ and no guard moves to $x_t$. So, as required, $x_t$ will not be occupied in the resultant configuration.

 If $c(x_t) \in C_2$, we need to show that $x_t$ is occupied with exactly one guard in the resultant configuration. 
 If $x_t \in C_2$, then the by hole-parent exchange, the guard on $x_t$ will not move, as $c(x_t) \in C_2$ and no other guard will come on $x_t$. Hence, $x_t$ is occupied with exactly one guard in the resultant configuration as required.
 If $x_t \notin C_2$, then by Observation~\ref{obs:f_in_c1}, $x_t \neq r$. Hence, the parent of $x_t$,  $p(x_t)$, would exist in $\Phi_3$. 
Further, since $C_2$ is a vertex cover by Lemma~\ref{lem:caterpillarproof}, $p(x_t) \in C_2$.  By hole-parent exchange, the only guard that moves to $x_t$ is from $p(x_t)$. So $x_t$ will be occupied with exactly one guard in the resultant configuration as required. 

\item \textbf{Case 2 :} $c(x_t) =x_j$. Note that $x_j \in C_2$ by our choice of $x_j$. Hence, showing $x_t$ is occupied with exactly one guard in the resultant configuration will suffice.  
We have $x_t \in V(\Phi_3)$ as $c(x_t) \neq r$. If $x_t \notin C_2$, then as discussed in Case $1$, $x_t \neq r$ and its parent, $p(x_t) \in V(\Phi_3) \cap C_2$. 
By hole-parent exchange, the only guard that moves to $x_t$  is from $p(x_t)$. If $x_t \in C_2$, then by hole-parent exchange, the guard on $x_t$ will not move because $c(x_t) \in C_2$. 
No other guard will come to $x_t$. Thus, as required, $x_t$ is occupied and has exactly one guard in the resultant configuration. 

\item \textbf{Case 3 :} $c(x_t) \in V(\Phi_2)$. 

If $c(x_t) \in C_2$, then by shift to parent movement, the guard from $c(x_t)$ will move to $x_t$ in the resultant configuration. Further, no other guard moves to $x_t$. Note that this holds even when $x_t =x_j$. 
If $x_t \in C_2 \cap V(\Phi_2)$, the guard on $x_t$ moves to its parent, whereas if $x_t = x_j$, the guard on $x_t$ moves to one of its cycle neighbours.
Thus, $x_t$ is occupied with exactly one guard in the resultant configuration.

If $c(x_t) \notin C_2$, then $x_t \in C_2$. By the guard movements, no guard will move to $x_t$. Further, if $x_t$ is $x_j$, the guard on $x_t$ will move to one of its cycle neighbours. Otherwise, the guard on $x_t$ will move to $p(x_t)$ by the shift to parent movement. Thus, as required, in the resultant configuration $x_t$ will have no guards. 

\item \textbf{Case 4 :} $c(x_t) =x_i$.

If $x_i\in C_2$, the guard from $x_i$ moves to its parent $x_t$ during the movement. Now, if $x_t =x_j$, by guards movements, the guard from $x_j$ moves to one of its cycle neighbours and no other guard moves to $x_j$ other than the guard from $c(x_t)=x_i$. If $x_t \neq x_j$, then $x_t \in \Phi_2$. By shift to parent movement, if $x_t \in C_2$, then the guard on $x_t$ moves to $p(x_t)$ and the only guard that comes to $x_t$ is from $c(x_t)=x_i$. If $x_t \notin C_2$, then also no other guard comes to $x_t$ other than the guard from $c(x_t)$. Thus, as needed, $x_t$ has exactly one guard in the resultant configuration.

If $c(x_t)=x_i \notin C_2$, then we need to show that $x_t$ is not occupied in the resultant  configuration.
Since $C_2$ is a vertex cover, $x_t \in C_2$. 
If $x_t = x_j$, by the guard movements, the guard from $x_j$ moves to one of its cycle neighbours and no guard moves to $x_j$. If $x_t \neq x_j$, then $x_t \in \Phi_2$. By shift to parent movement,  the guard on $x_t$ moves to $p(x_t)$ and no other guard comes to $x_t$. Thus, as needed, $x_t$ is not occupied in the resultant configuration.

Note that these arguments go through even when $x_i =f$.

\item \textbf{Case 5 :} $c(x_t) \in V(\Phi_1)$. Recall that by definition of $\Phi_1$, $\Phi_1 = \emptyset$ whenever $x_i =f$. Hence, $x_i \neq f$. By our choice of $x_i$, $x_i \notin C_2$. Now, there are two possibilities: either $x_t = x_i$ or $x_t \in V(\Phi_1)$.

If $x_t = x_{i}$ then $c(x_i) \in C_2 $.
In this case, we want to show that $x_i$ has exactly one guard in the resultant configuration.
Now, by the guard movements defined in the lemma, $x_i$ gets a guard from one of its cycle neighbours. Also, there is no guard movement from $c(x_i)$ or $p(x_i)$ to $x_i$. Hence $x_i$ is occupied by exactly one guard in the resultant configuration. 

Now, consider the case when $x_t \in \Phi_1$. If $c(x_t) \in C_2$ and $x_t \in C_2$ then hole-parent exchange movement, no guard moves to $x_t$ and the guard on $x_t$ remains on $x_t$.
If $c(x_t) \in C_2$ and $x_t \notin C_2$, then $p(x_t) \in C_2$. This implies $p(x_t) \neq x_i$. Now, by hole-parent exchange movement, $x_t$ gets a guard from its parent and no other guard comes to $x_t$.
Hence, whenever $c(x_t) \in C_2$, $x_t$ gets exactly one guard in the resultant configuration after the movements.  If $c(x_t) \notin C_2$, then $x_t \in C_2$. By the hole-parent exchange movement, the guard from $x_t$ moves to $c(x_t)$ and no other guard comes to $x_t$. Hence, as required, $x_t$ has no guard in the resultant configuration.

Note that these arguments go through even when $c(x_t)=f$.

 \end{itemize}
Hence, in all cases, the resultant guard positions on the spine are the same as those in $C_1$ given by Algorithm~\ref{algo:main}.

\end{proof}
\begin{proof}[Proof of Lemma~\ref{lem:spiltcyclemove}]
Let $C_1$ be the starting configuration. Suppose that we perform the movements as mentioned in the first part of the lemma. 
Since the paths $\Psi_1$ and $\Psi_2$ are vertex disjoint and vertices $c_i$ and $c_j$ do not lie on them, no guard is moved more than once. 
We need to prove that a vertex $c_t \in V(C)$ is occupied in the resultant configuration if and only if $c_t \in C_2$.  
We will also show that in the resultant configuration, every vertex has at most one guard.

By Algorithm~\ref{algo:main},  $c_{t} \in C_2$  if and only if $c_{t-1} \in C_1$.
Hence, to show that $c_t$ is occupied in the resultant configuration if and only if $c_t \in C_2$, it suffices to show that $c_{t-1} \in C_1$ if and only if $c_{t}$ is in the resultant configuration. 

Observe that by construction of $C_1$, vertices $c_{i-1},c_i$ and $c_{i+1}$ cannot be simultaneously occupied in $C_1$. Similarly, vertices $c_{j-1},c_j$ and $c_{j+1}$ cannot be simultaneously occupied in $C_1$. Further, note that $c_t$ can get a guard only from $c_{t-1},c_{t+1}$ or from its spine neighbour (the last case is possible only when $c_t = c_i$).
\begin{itemize}
    \item \textbf{Case 1:} $c_{t-1}=c_i$.
    
    By assumption, $c_i \notin C_1$. Hence, it suffices to show that $c_t =c_{i+1}$ is not occupied in the resultant configuration. Since $C_1$ is a vertex cover by Lemma~\ref{lem:caterpillarproof}, $c_{i+1} \in C_1$. Since $c_i \notin C_1$, $c_{i+1}$ can get a guard only from $c_{i+2}$.
    
    If $c_{i+1} = c_j$, then the path $\Psi_1 = \emptyset$. By the guard movements defined in the lemma, the guard from $c_j$ moves to its spine neighbour.
    Now, $c_{j+1} \in V(\Psi_2)$ and a restricted clockwise shift is performed on $\Psi_2$. 
    Since $c_j \in C_1$, even when $c_{j+1} \in C_1$, no guard will move to $c_j$. Hence, in the resultant configuration $c_{t} =c_j$ will have no guard.
    
    If $c_{i+1} \neq c_j$, then $c_{i+1} \in V(\Psi_1)$. Observe that $c_{i+2}$ is either in $V(\Psi_1)$ or $c_{i+2} = c_j$. Therefore, no guard moves from $c_{i+2}$ to $c_{i+1}$. By performing an anticlockwise shift on $\Psi_1$, the guard from $c_{i+1}$ will move to $c_{i+2}$. 
    Hence, $c_{t}=c_{i+1}$ will be unoccupied in the resultant configuration.
    
    \item \textbf{Case 2:} $c_{t-1} \in V(\Psi_1)$. 
    In this case there are two possibilities: either $c_t \in V(\Psi_1)$ or $c_t = c_j$.

    When $c_{t} = c_j$, by the guard movements defined in the lemma, the guard from $c_j$ moves to its spine neighbour. Now, either $\Psi_2 = \emptyset$ or $c_{j+1} \in V(\Psi_2)$. If $\Psi_2 = \emptyset$ then $c_{j+1} = c_i$.
    Since $c_i = c_{j+1} \notin C_1$ by our choice of $c_i$, no guard will move from $c_{j+1}$ to $c_{j}$. If $c_{j+1} \in V(\Psi_2)$, then  no guard moves from $c_{j+1}$ to $c_j$ as a restricted clockwise shift is applied on $\Psi_2$ and $c_j \in C_1$. Hence, a guard can come to $c_j$ only from $c_{j-1} = c_{t-1}$.
     An anticlockwise shift is applied on $\Psi_1$.
     If $c_{j-1} \in C_1$, then the guard from $c_{j-1}$ moves to $c_{j} = c_{t}$.  Hence, as required, $c_{t}=c_j$ has exactly one guard on it in the resultant configuration. 
    If $c_{j-1} \notin C_1$, on applying anticlockwise shift on $\Psi_1$, no guard moves from $c_{j-1}$  to $c_{j}$. Hence, as required,  $c_j = c_{t}$ has no guard in the resultant configuration.

    Now, suppose that $c_{t} \in V(\Psi_1)$. Notice that $c_{t+1} \notin V(\Psi_2)$ and so no guard can move from $c_{t+1}$ to $c_t$ by the movements defined in the lemma. If $c_{t-1} \in C_1$, then on performing the anticlockwise shift on $\Psi_1$, the guard on $c_{t-1}$ moves to $c_{t}$. If $c_{t} \in C_1$, then the guard from $c_{t}$ will move to $c_{t+1}$. Hence, as required, $c_{t}$ will have exactly one guard in the resultant configuration. If $c_{t-1} \notin C_1$, then $c_{t} \in C_1$. On performing the anticlockwise shift, the guard on $c_{t}$ will move to $c_{t+1}$ and no other guard will come on $c_{t}$. Hence, as required, $c_{t}$ will be unoccupied in the resultant configuration.

    \item \textbf{Case 3:} $c_{t-1} = c_j$. 

    By assumption, $c_j \in C_1$. Hence, it suffices to show that $c_t =c_{j+1}$ is occupied with exactly one guard in the resultant configuration. Since, the guard on $c_j$ moves to its spine neighbour, $c_{j+1}$ can get guard only from its spine neighbour (when $c_{j+1} =c_i$) or from $c_{j+2}$. 
    
    Now, there are two possibilities: either $c_t =c_{j+1}=c_i$ or $c_t \in V(\Psi_2)$.
     If $c_{j+1}=c_i$, then by our choice, $c_i \notin C_1$. By the guard movements defined in the lemma, a guard from spine is moved to $c_i$. Since, $c_{j+2}=c_{i+1} \in V(\Psi_1)$, with an anticlockwise shift on $\Psi_1$, it will not contribute any guard to $c_{j+1}=c_i$. Hence, as required, $c_i = c_{t}$ has exactly one guard in the resultant configuration.

     If $c_{j+1} \neq c_i$, then $c_t = c_{j+1} \in V(\Psi_2$). In this case, $c_{j+1}=c_t$ can get a guard only from $c_{j+2}$. If $c_{j+1} \in C_1$, then by restricted clockwise shift movement on $\Psi_2$, the guard from $c_{j+1}$ does not move as $c_{j} \in C_1$. Similarly, no guard moves from  $c_{j+2}$ to $c_{j+1}$ as $c_{j+1} \in C_1$. 
     Hence, as required, there is exactly one guard on $c_{j+1}=c_{t}$ in the resultant configuration.
     If $c_{j+1} \notin C_1$, then $c_{j+2} \in C_1$. Therefore, by our choice of $c_i$, $c_{j+2} \neq c_i$. Hence, $c_{j+2} \in V(\Psi_2)$. By the restricted clockwise shift movement on $\Psi_2$, the guard from $c_{j+2}$ will move to $c_{j+1}$. Hence, as required, $c_{j+1}=c_{t}$ will have exactly one guard in the resultant configuration.
    
    \item \textbf{Case 4:} $c_{t-1} \in V(\Psi_2)$.
    
 Notice that no guard moves from $c_{t-1}$ to $c_t$, as guard movement on $\Psi_2$ is restricted clockwise shift. Hence, $c_t$ can get a guard only from its spine neighbour (when $c_t=c_i$) or from $c_{t+1}$. 
 
 There are two possibilities: either $c_t \in V(\Psi_2)$ or $c_t = c_i$.

    When $c_{t} = c_i$, by our choice of $c_i$, $c_t= c_i \notin C_1$. This implies that $c_{t-1} \in C_1$. Hence,in this case, we need to show that $c_t$ gets exactly one guard in the resultant configuration.
    Now, $c_{t+1} =c_{i+1} \notin \Psi_2$. Hence, by the guard movements defined in the lemma, no guard moves from $c_{i+1}$ to $c_{i}$ and
     $c_i$ gets a guard from its spine neighbour. 
    Hence $c_i=c_{t}$ has exactly one guard in the resultant configuration. 

    Now consider the case when $c_t \in  V(\Psi_2)$.  
    If $c_{t-1} \in C_1$ and $c_{t} \in C_1$, then by our choice of $C_1$, $c_{t+1} \notin C_1$. Hence no guard can come to $c_t$ from $c_{t+1}$. Further, on performing the restricted clockwise shift on $\Psi_2$, the guard on $c_{t}$ will not move. Hence, as required, in the resultant configuration, $c_{t}$ has exactly one guard.
    If $c_{t-1} \in C_1$ and $c_{t} \notin C_1$, then $c_{t+1} \in C_1$.
    By our choice of $c_i$, $c_{t+1} \neq c_i$ and so $c_{t+1} \in \Psi_2$.
    Further, on performing the restricted clockwise shift on $\Psi_2$, the guard on $c_{t+1}$ will move $c_{t}$.  Hence, as required, in the resultant configuration, $c_{t}$ has exactly one guard.
    If $c_{t-1} \notin C_1$, then $c_{t} \in C_1$. Further, on performing the restricted clockwise shift, the guard on $c_{t}$ will move to $c_{t-1}$. Also, since $c_t \in C_1$, no guard will move from $c_{t+1}$ to $c_{t}$. Hence, as required, in the resultant configuration, $c_{t}$ has no guard.

\end{itemize}
Hence, in all cases, the resultant guard positions on the vertices of the cycle are the same as those in $C_2$ given by Algorithm~\ref{algo:main}. This completes the proof of the first part of the lemma.

For the second part of the lemma, the proof is symmetric with the following role exchanges: 
\begin{itemize}
    \item $C_1$ with $C_2$
    \item $\Psi_1$ with $\Psi_2$
    \item  anticlockwise shift with clockwise shift 
    \item  restricted anticlockwise shift with restricted clockwise shift 
    \item  $c_{i-1}$ with $c_{i+1}$ and  $c_{i-2}$ with $c_{i+2}$  
    \item   $c_{j-1}$ with $c_{j+1}$ and  $c_{j-2}$ with $c_{j+2}$
    \item   $c_{t-1}$ with $c_{t+1}$ and  $c_{t-2}$ with $c_{t+2}$.
\end{itemize}

For the sake of completeness, we now give the complete proof of the second part of the lemma. The reader may choose to skip this section of the proof.

Let $C_2$ be the starting configuration. Suppose that we perform the movements as mentioned in the second part of the lemma. 
Since the paths $\Psi_1
$ and $\Psi_2$ are vertex disjoint and vertices $c_i$ and $c_j$ do not lie on them, no guard is moved more than once. 
We need to prove that a vertex $c_t \in V(C)$ is occupied in the resultant configuration if and only if $c_t \in C_1$.  
We will also show that in the resultant configuration, every vertex has at most one guard.

By Algorithm~\ref{algo:main},  $c_{t} \in C_1$  if and only if $c_{t+1} \in C_2$.
Hence, to show that $c_t$ is occupied in the resultant configuration if and only if $c_t \in C_1$, it suffices to show that $c_{t+1} \in C_2$ if and only if $c_{t}$ is in the resultant configuration. 

Observe that by construction of $C_2$, vertices $c_{i-1},c_i$ and $c_{i+1}$ cannot be simultaneously occupied in $C_2$. Similarly, vertices $c_{j-1},c_j$ and $c_{j+1}$ cannot be simultaneously occupied in $C_2$. Further, note that $c_t$ can get a guard only from $c_{t-1},c_{t+1}$ or from its spine neighbour (the last case is possible only when $c_t = c_i$).
\begin{itemize}
    \item \textbf{Case 1:} $c_{t+1}=c_i$.
    
    By assumption, $c_i \notin C_2$. Hence, it suffices to show that $c_t =c_{i-1}$ is not occupied in the resultant configuration. Since $C_2$ is a vertex cover by Lemma~\ref{lem:caterpillarproof}, $c_{i-1} \in C_2$. Since $c_i \notin C_2$, $c_{i-1}$ can get a guard only from $c_{i-2}$.
    
    If $c_{i-1} = c_j$, then the path $\Psi_2 = \emptyset$. By the guard movements defined in the lemma, the guard from $c_j$ moves to its spine neighbour.
    Now, $c_{j-1} \in V(\Psi_1)$ and a restricted anticlockwise shift is performed on $\Psi_1$. 
    Since $c_j \in C_2$, even when $c_{j-1} \in C_2$, no guard will move to $c_j$. Hence, in the resultant configuration $c_{t} =c_j$ will have no guard.
    
    If $c_{i-1} \neq c_j$, then $c_{i-1} \in V(\Psi_2)$. Observe that $c_{i-2}$ is either in $V(\Psi_2)$ or $c_{i-2} = c_j$. Therefore, no guard moves from $c_{i-2}$ to $c_{i-1}$. By performing a clockwise shift on $\Psi_2$, the guard from $c_{i-1}$ will move to $c_{i-2}$. 
    Hence, $c_{t}=c_{i-1}$ will be unoccupied in the resultant configuration.
    
    \item \textbf{Case 2:} $c_{t+1} \in V(\Psi_2)$. 
    In this case there are two possibilities: either $c_t \in V(\Psi_2)$ or $c_t = c_j$.

    When $c_{t} = c_j$, by the guard movements defined in the lemma, the guard from $c_j$ moves to its spine neighbour. Now, either $\Psi_1 = \emptyset$ or $c_{j-1} \in V(\Psi_1)$. If $\Psi_1 = \emptyset$ then $c_{j-1} = c_i$.
    Since $c_i= c_{j-1} \notin C_2$ by our choice of $c_i$, no guard will move from $c_{j-1}$ to $c_{j}$. If $c_{j-1} \in V(\Psi_1)$, then  no guard moves from $c_{j-1}$ to $c_j$ as a restricted anticlockwise shift is applied on $\Psi_1$ and $c_j \in C_2$. Hence, a guard can come to $c_j$ only from $c_{j+1} = c_{t+1}$.
     A clockwise shift is applied on $\Psi_2$.
     If $c_{j+1} \in C_2$, then the guard from $c_{j+1}$ moves to $c_{j} = c_{t}$.  Hence, as required, $c_{t}=c_j$ has exactly one guard on it in the resultant configuration. 
    If $c_{j+1} \notin C_2$, on applying clockwise shift on $\Psi_2$, no guard moves from $c_{j+1}$  to $c_{j}$. Hence, as required,  $c_j = c_{t}$ has no guard in the resultant configuration.

    Now, suppose that $c_{t} \in V(\Psi_2)$. Notice that $c_{t-1} \notin V(\Psi_1)$ and so no guard can move from $c_{t-1}$ to $c_t$ by the movements defined in the lemma. If $c_{t+1} \in C_2$, then on performing the clockwise shift on $\Psi_2$, the guard on $c_{t+1}$ moves to $c_{t}$. If $c_{t} \in C_2$, then the guard from $c_{t}$ will move to $c_{t-1}$. Hence, as required, $c_{t}$ will have exactly one guard in the resultant configuration. If $c_{t+1} \notin C_2$, then $c_{t} \in C_2$. On performing the clockwise shift, the guard on $c_{t}$ will move to $c_{t-1}$ and no other guard will come on $c_{t}$. Hence, as required, $c_{t}$ will be unoccupied in the resultant configuration.

    \item \textbf{Case 3:} $c_{t+1} = c_j$. 

    By assumption, $c_j \in C_2$. Hence, it suffices to show that $c_t =c_{j-1}$ is occupied with exactly one guard in the resultant configuration. Since, the guard on $c_j$ moves to its spine neighbour, $c_{j-1}$ can get guard only from its spine neighbour (when $c_{j-1} =c_i$) or from $c_{j-2}$. 
    
    Now, there are two possibilities: either $c_t =c_{j-1}=c_i$ or $c_t \in V(\Psi_1)$.
     If $c_{j-1}=c_i$, then by our choice, $c_i \notin C_2$. By the guard movements defined in the lemma, a guard from spine is moved to $c_i$. Since, $c_{j-2}=c_{i-1} \in V(\Psi_2)$, with a clockwise shift on $\Psi_2$, it will not contribute any guard to $c_{j-1}=c_i$. Hence, as required, $c_i = c_{t}$ has exactly one guard in the resultant configuration.

     If $c_{j-1} \neq c_i$, then $c_t = c_{j-1} \in V(\Psi_1$). In this case, $c_{j-1}=c_t$ can get a guard only from $c_{j-2}$. If $c_{j-1} \in C_2$, then by restricted anticlockwise shift movement on $\Psi_1$, the guard from $c_{j-1}$ does not move as $c_{j} \in C_2$. Similarly, no guard moves from  $c_{j-2}$ to $c_{j-1}$ as $c_{j-1} \in C_2$. 
     Hence, as required, there is exactly one guard on $c_{j-1}=c_{t}$ in the resultant configuration.
     If $c_{j-1} \notin C_2$, then $c_{j-2} \in C_2$. Therefore, by our choice of $c_i$, $c_{j-2} \neq c_i$. Hence, $c_{j-2} \in V(\Psi_1)$. By the restricted anticlockwise shift movement on $\Psi_1$, the guard from $c_{j-2}$ will move to $c_{j-1}$. Hence, as required, $c_{j-1}=c_{t}$ will have exactly one guard in the resultant configuration.
    
    \item \textbf{Case 4:} $c_{t+1} \in V(\Psi_1)$.
    
 Notice that no guard moves from $c_{t+1}$ to $c_t$, as guard movement on $\Psi_1$ is restricted anticlockwise shift. Hence, $c_t$ can get a guard only from its spine neighbour (when $c_t=c_i$) or from $c_{t-1}$. 
 
 There are two possibilities: either $c_t \in V(\Psi_1)$ or $c_t = c_i$.

    When $c_{t} = c_i$, by our choice of $c_i$, $c_t= c_i \notin C_2$. This implies that $c_{t+1} \in C_2$. Hence, in this case, we need to show that $c_t$ gets exactly one guard in the resultant configuration.
    Now, $c_{t-1} =c_{i-1} \notin \Psi_1$. Hence, by the guard movements defined in the lemma, no guard moves from $c_{i-1}$ to $c_{i}$ and
     $c_i$ gets a guard from its spine neighbour. 
    Hence $c_i=c_{t}$ has exactly one guard in the resultant configuration. 

    Now consider the case when $c_t \in  V(\Psi_1)$.  
    If $c_{t+1} \in C_2$ and $c_{t} \in C_2$, then by our choice of $C_2$, $c_{t-1} \notin C_2$. Hence no guard can come to $c_t$ from $c_{t-1}$. Further, on performing the restricted anticlockwise shift on $\Psi_1$, the guard on $c_{t}$ will not move. Hence, as required, in the resultant configuration, $c_{t}$ has exactly one guard.
    If $c_{t+1} \in C_2$ and $c_{t} \notin C_2$, then $c_{t-1} \in C_2$.
    By our choice of $c_i$, $c_{t-1} \neq c_i$ and so $c_{t-1} \in \Psi_1$.
    Further, on performing the restricted anticlockwise shift on $\Psi_1$, the guard on $c_{t-1}$ will move $c_{t}$.  Hence, as required, in the resultant configuration, $c_{t}$ has exactly one guard.
    If $c_{t+1} \notin C_2$, then $c_{t} \in C_2$. Further, on performing the restricted anticlockwise shift, the guard on $c_{t}$ will move to $c_{t+1}$. Also, since $c_t \in C_2$, no guard will move from $c_{t-1}$ to $c_{t}$. Hence, as required, in the resultant configuration, $c_{t}$ has no guard.

\end{itemize}
Hence, in all cases, the resultant guard positions on the vertices of the cycle are the same as those in $C_1$ given by Algorithm~\ref{algo:main}.

\end{proof}
\begin{proof}[Proof of Lemma~\ref{lem:size_of_c1}]
We show this by comparing the sizes of the minimum vertex cover $S$ of $G$ used in Algorithm~\ref{algo:main} and  the configuration $C_1$ obtained from Algorithm~\ref{algo:main}.
          By construction,
    $C_1= P_1 \uplus  Q_1 =  ((S \cap P(T)) \cup \{f\}  \cup F ) \uplus (S \cap V(C))$. Further, from the algorithm, it follows that $F \cap (S \cup \{f\}) = \emptyset$. Therefore $C_1 = (S \cup \{f\}) \uplus F$.
     To show that $|C_1| \leq \frac{4}{3}|S|$, it suffices to map each vertex of $C_1 \setminus S$ to three distinct vertices from $C_1 \cap S$ such that no vertex from $S$ is mapped more than once.  
     
     Recall that $H = P(T)\setminus S $ and $F= F_1 \cup F_2$, where  $F_1 = \{ v \mid v \in H$, $v \neq f,r$ and  $v$ has multiple neighbours on $C \}$ and
     $F_2 = \{ v \mid v \in H$, $v \neq f$ and the child of $v$ in the spine has a neighbour in $S \cap V(C) \}$.

    First, consider a vertex $v \in F_1$. Since $v \notin S $, by definition of $H$, all neighbours of $v$ are in $S$.  Since $v \neq f$, the vertex $v$ has a child $w$ on the spine. Now by definition of $F_1$, $v$ has at least two cycle neighbours, say $x_i$ and $x_j$. Map $ v \longmapsto \{w,x_i,x_j\}$.
    Next, let $v \in F_2 \setminus F_1$. Observe that $v$ has at least one cycle neighbour $x_i$ that is in $S$. Since $v \neq f$, the vertex $v$ has a child $w$ on the spine. Also, by definition of $F_2$, $w$ has a neighbour in $S \cap V(C)$. Let this vertex be $x_j$. In this case, we can map  $v \longmapsto \{w,x_i,x_j\}$. In particular, if the root vertex $r$ is in $F_2$, the same mapping rule would apply. However, it is important to notice that in all cases, there will be at least one occupied cycle neighbour of $r$ that is not used for mapping to any vertex in $F=F_1 \cup F_2$.
Now if $f \in S$, we have already mapped all vertices of $C_1 \setminus S$ as required. If $f \notin S$, then map $f$ to two of its cycle neighbours and a cycle neighbour of $r$ that is in $S$ which is unmapped so far. It can be  verified that no vertex from $S$ was used more than once in the mapping.

This shows that $|C_1 \setminus S| \leq \frac{|S|}{3}$ and hence $|C_1| \leq \frac{4}{3}|S|$.
This completes the proof.
\end{proof}


\begin{thebibliography}{10}
\providecommand{\url}[1]{\texttt{#1}}
\providecommand{\urlprefix}{URL }
\providecommand{\doi}[1]{https://doi.org/#1}

\bibitem{Hisashi_Gen_Trees}
Araki, H., Fujito, T., Inoue, S.: On the eternal vertex cover numbers of generalized trees. IEICE Transactions on Fundamentals of Electronics, Communications and Computer Sciences  \textbf{E98.A},  1153--1160 (06 2015)

\bibitem{BABU2021}
Babu, J., Chandran, L.S., Francis, M., Prabhakaran, V., Rajendraprasad, D., Warrier, N.J.: On graphs whose eternal vertex cover number and vertex cover number coincide. Discrete Applied Mathematics  \textbf{https://doi.org/10.1016/j.dam.2021.02.004} (2021)

\bibitem{bgs2026evchalinarxiv}
Babu, J., Ghosal, P., Simoes, C.: Eternal vertex cover problem on halin graphs (2026), \url{https://arxiv.org/abs/2607.23155}

\bibitem{BabuKPW25}
Babu, J., Krishnan, K.M., Prabhakaran, V., Warrier, N.J.: Computing eternal vertex cover number of maximal outerplanar graphs in linear time. Theor. Comput. Sci.  \textbf{1056},  115530 (2025). \doi{10.1016/J.TCS.2025.115530}, \url{https://doi.org/10.1016/j.tcs.2025.115530}

\bibitem{misra2022eternal}
Babu, J., Misra, N., Nanoti, S.: Eternal vertex cover on bipartite graphs. In: Kulikov, A.S., Raskhodnikova, S. (eds.) Computer Science - Theory and Applications - 17th International Computer Science Symposium in Russia, {CSR} 2022, Virtual Event, June 29 - July 1, 2022, Proceedings. Lecture Notes in Computer Science, vol. 13296, pp. 64--76. Springer (2022). \doi{10.1007/978-3-031-09574-0\_5}, \url{https://doi.org/10.1007/978-3-031-09574-0\_5}

\bibitem{BABU22newlb}
Babu, J., Prabhakaran, V.: A new lower bound for the eternal vertex cover number of graphs. J. Comb. Optim.  \textbf{44}(4),  2482--2498 (2022). \doi{10.1007/S10878-021-00764-8}, \url{https://doi.org/10.1007/s10878-021-00764-8}

\bibitem{BPS2021}
Babu, J., Prabhakaran, V., Sharma, A.: A substructure based lower bound for eternal vertex cover number. Theoretical Computer Science  \textbf{890} (2021)

\bibitem{Bodlaender1996tw}
Bodlaender, H.L.: A linear-time algorithm for finding tree-decompositions of small treewidth. SIAM Journal on Computing  \textbf{25},  1305--1317 (1996)

\bibitem{Bodlaender98}
Bodlaender, H.L.: A partial \emph{k}-arboretum of graphs with bounded treewidth. Theor. Comput. Sci.  \textbf{209}(1-2),  1--45 (1998). \doi{10.1016/S0304-3975(97)00228-4}, \url{https://doi.org/10.1016/S0304-3975(97)00228-4}

\bibitem{BorieParkerTovey2008}
Borie, R.B., Parker, R.G., Tovey, C.A.: Solving problems on recursively constructed graphs. ACM Computing Surveys  \textbf{41}(1),  4:1--4:51 (2008). \doi{10.1145/1456650.1456654}

\bibitem{TizianaseriesParallel}
Calamoneri, T., Cor{\`{o}}, F., Paesani, G.: The minimum eternal vertex cover problem on a subclass of series-parallel graphs. CoRR  \textbf{abs/2504.04897} (2025). \doi{10.48550/ARXIV.2504.04897}, \url{https://doi.org/10.48550/arXiv.2504.04897}

\bibitem{cockayne2004roman}
Cockayne, E.J., Dreyer~Jr, P.A., Hedetniemi, S.M., Hedetniemi, S.T.: Roman domination in graphs. Discrete mathematics  \textbf{278}(1-3),  11--22 (2004)

\bibitem{Cornuejols1983}
Cornu{\'e}jols, G., Naddef, D., Pulleyblank, W.R.: {H}alin graphs and the travelling salesman problem. Mathematical Programming  \textbf{26}(3),  287--294 (1983)

\bibitem{Eppstein1988}
Eppstein, D.: Parallel ${O}(\log{n})$ time edge-colouring of trees and {H}alin graphs. Information Processing Letters  \textbf{27}(1),  43--51 (1988)

\bibitem{Eppstein2016}
Eppstein, D.: Simple recognition of {H}alin graphs and their generalizations. Journal of Graph Algorithms and Applications  \textbf{20},  323--346 (02 2016). \doi{10.7155/jgaa.00395}

\bibitem{Fomin2010}
Fomin, F.V., Gaspers, S., Golovach, P.A., Kratsch, D., Saurabh, S.: Parameterized algorithm for eternal vertex cover. Information Processing Letters  \textbf{110}(16),  702--706 (2010)

\bibitem{FominT06}
Fomin, F.V., Thilikos, D.M.: A 3-approximation for the pathwidth of {H}alin graphs. J. Discrete Algorithms  \textbf{4}(4),  499--510 (2006). \doi{10.1016/J.JDA.2005.06.004}, \url{https://doi.org/10.1016/j.jda.2005.06.004}

\bibitem{goddard2005}
Goddard, W., Hedetniemi, S.M., Hedetniemi, S.T.: Eternal security in graphs. J. Combin. Math. Combin. Comput  \textbf{52},  169--180 (2005)

\bibitem{Halin1971orig}
Halin, R.: Studies on minimally n-connected graphs. Combinatorial Mathematics and its Applications pp. 129--136 (1971)

\bibitem{Hartnell2014}
Hartnell, B., Mynhardt, C.: Independent protection in graphs. Discrete Mathematics  \textbf{335},  100 -- 109 (2014)

\bibitem{Hazewinkel1997}
Hazewinkel, M. (ed.): Encyclopaedia of Mathematics: Supplement Volume I. Encyclopaedia of Mathematics, Springer, Dordrecht, 1 edn. (1997). \doi{10.1007/978-94-015-1288-6}

\bibitem{Karp72}
Karp, R.M.: Reducibility among combinatorial problems. In: Miller, R.E., Thatcher, J.W. (eds.) Proceedings of a symposium on the Complexity of Computer Computations, held March 20-22, 1972, at the {IBM} Thomas J. Watson Research Center, Yorktown Heights, New York, {USA}. The {IBM} Research Symposia Series, vol.~40, pp. 85--103. Plenum Press, New York (1972)

\bibitem{Klostermeyer2009}
Klostermeyer, W.F., Mynhardt, C.M.: Edge protection in graphs. Australasian Journal of Combinatorics  \textbf{45},  235--250 (2009)

\bibitem{klostermeyer2017eternal}
Klostermeyer, W.F., MacGillivray, G.: Eternal domination: Criticality and reachability. Discussiones Mathematicae Graph Theory  \textbf{37}(1),  63--77 (2017)

\bibitem{Syslo1979halin}
Sys{\l}o, M.M.: On the {H}amiltonian properties of {H}alin graphs. Discrete Mathematics  \textbf{27}(2),  215--220 (1979). \doi{10.1016/0012-365X(79)90071-2}

\bibitem{Vazirani}
Vazirani, V.V.: Approximation algorithms. Springer (2001)

\end{thebibliography}
\end{document}